\documentclass[sigconf, 10pt, screen]{acmart}
\usepackage{xcolor}
\usepackage{multirow}
\usepackage{url}
\usepackage{array}

\usepackage{enumitem}
\usepackage{float}
\usepackage{subfig}
\usepackage{fontawesome}
\usepackage[linesnumbered,ruled,vlined]{algorithm2e}
\usepackage{algpseudocode}
\usepackage{hhline}
\usepackage{pifont}
\usepackage{bbding}
\usepackage{makecell}
\usepackage{graphicx}
\usepackage{amsmath}
\usepackage{amsfonts}
\usepackage{hyperref}
\usepackage{cleveref}
\usepackage{soul}
\usepackage{colortbl}
\usepackage[switch]{lineno}
\usepackage{tcolorbox}
\tcbuselibrary{listings,breakable}

\hypersetup{
  colorlinks   = true,    
  urlcolor     = purple,    
  linkcolor    = red,    
  citecolor    = green   
}
\usepackage{color}
\definecolor{green}{rgb}{0, 0.5, 0}
\definecolor{orange}{rgb}{0.8, 0.6, 0.2}
\definecolor{orange2}{rgb}{1.0, 0.6, 0.2}
\definecolor{red}{rgb}{1.0, 0.0, 0.0}
\definecolor{teal}{rgb}{0.0, 0.4, 0.4}
\definecolor{purple}{rgb}{0.65,0,0.65}
\definecolor{saffron}{rgb}{0.95,0.75,0.2}
\definecolor{turquoise}{rgb}{0.0,0.5,0.5}
\definecolor{black}{rgb}{0.0, 0.0, 0.0}
\definecolor{gray}{rgb}{0.5, 0.5, 0.5}

\AtBeginDocument{%
  }

\begin{document}

\setcopyright{cc}
\copyrightyear{2018}
\acmYear{2018}
\acmDOI{XXXXXXX.XXXXXXX}
\acmConference[Conference acronym 'XX]{Make sure to enter the correct
  conference title from your rights confirmation email}{June 03--05,
  2018}{Woodstock, NY}
\acmISBN{978-1-4503-XXXX-X/2018/06}

\title{MotionQ: Operator-Conditioned Motion Quotients for Cross-Observation WiFi Gesture Recognition}

\author{
Xiang Zhang$^{1}$,
Huan Yan$^{2}$,
Geying Yang$^{1}$,
Jianchun Liu$^{3}$,\\
Tao Liu$^{4}$,
Zhi Liu$^{5}$,
Meng Li$^{6}$,
}

\affiliation{
$^{1}$ Tianjin University \country{}
$^{2}$ Guizhou Normal University \country{}
$^{3}$ University of Science and Technology of China \country{}
$^{4}$ Guangzhou University \country{}
$^{5}$ The University of Electro-Communications \country{}
$^{6}$ Hefei University of Technology \country{}
}

\renewcommand{\shortauthors}{XX et al.}

\begin{abstract}
WiFi gesture recognition is accurate in fixed deployments but often degrades when user orientation, available links, or transceiver placement changes. Unlike ordinary domain shifts, these changes alter the wireless observation operator, so the same motion is expected to produce different measurements. Existing methods nevertheless pursue domain-invariant features and largely overlook changing layouts and observation configurations. Yet changing the observation operator also changes which task-relevant motion cues are physically observable, rather than merely altering the appearance of a fixed set of cues. Under a local linearization of the WiFi forward process, we derive a common task-observability condition under which a strict common linear representation is recoverable from every geometry-induced operator while preserving the gesture task. When the condition fails, enforcing stronger alignment across additional heterogeneous source operators may discard task-relevant cues still observable under individual operators. We therefore present MotionQ, which generates an operator-conditioned two-support motion measure for each candidate geometry. A motion quotient removes only the arbitrary ordering of its unlabeled supports and is represented by permutation-invariant central moments. Rather than matching quotients across operators, single-link-retention interventions encourage each view to retain information sufficient for gesture recognition. Extensive evaluations show that MotionQ is robust to extrapolative observation operators. 

\end{abstract}




\maketitle

\section{Introduction}
\label{sec:intro}

Gesture recognition plays an important role in human computer interaction~\cite{kim2026simplified,wang2025vr}, and WiFi-based gesture recognition~\cite{li2025cross,yan2025wi} has attracted growing interest due to the ubiquity of wireless infrastructure, the emergence of next-generation integrated sensing and communication (ISAC)~\cite{9363693,hu2025poison}, and its non-intrusive sensing capability~\cite{li2024uwb,fan2026sense}. However, WiFi sensing often suffers from severe generalization challenges. Unlike vision systems that directly observe body motion, WiFi senses gestures only indirectly through motion-induced perturbations to wireless propagation paths, resulting in relatively coarse and strongly geometry-dependent observations~\cite{li2025wilife,zhang2026beyond}. Consequently, its sensing patterns are highly sensitive to changes in the environment, user motion, and the geometric relationship among the user, transmitter, and receiver. A substantial body of prior work has therefore focused on improving the generalizability of WiFi sensing across such variations~\cite{miao2025wi,li2025cross1,chen2024wignn,gao2022towards,zhang2026wi}.

Existing efforts to improve cross-domain WiFi gesture recognition can be broadly divided into domain adaptation and domain generalization~\cite{chen2023cross}. Domain adaptation transfers a source model to a new deployment using samples collected from the target domain. Few-shot approaches reduce this burden to only a small number of target samples~\cite{yin2022fewsense,sheng2024metaformer,feng2022wi,liu2025efficient}, but the overall data collection cost remains substantial in complex WiFi environments. Domain generalization instead aims to recognize gestures in unseen domains without access to target-domain data. One line of work explicitly constructs physically motivated domain-independent features~\cite{chen2026wipihgr}, such as BVP~\cite{zhang2021widar3} and MNP~\cite{gao2021towards}. However, such handcrafted abstractions inevitably retain only the motion characteristics anticipated by their design and may therefore discard useful information. A second line lets neural networks discover invariant representations through mechanisms such as spatial-temporal attention, style randomization, cross-view consistency, or adversarial domain alignment~\cite{gu2022wigrunt,zhang2026beyond,liu2023wisr,liu2024unifi,zhang2026wi,wang2022airfi,wei2025source}. These approaches share a common goal: suppress domain-dependent variations while preserving representations that remain stable across domains, resembling the treatment of domain or style variations in computer vision. This abstraction can be effective for style variations such as hardware distortions or environmental backgrounds, but becomes questionable when user orientation, available links, or transceiver layouts change, because these factors alter how motion is physically observed rather than merely how the observation appears. Cross-layout sensing is particularly important in practical deployments but has received comparatively limited attention. A recent effort, PerceptAlign~\cite{jia2026breaking}, studies cross-layout WiFi-based 3D pose estimation, but mainly treats transceiver geometry as an explicit conditioning feature. We refer to this broader deployment problem as cross-observation generalization: the ability to recognize gestures when the wireless observation operator changes, i.e., the geometric relationship through which WiFi transceivers observe gestures. This raises a question: 
\begin{tcolorbox}[colback=gray!25!white, size=title, breakable, boxsep=1mm, colframe=white, before={\vskip1mm}, after={\vskip0mm}]
\textbf{Q1:} Can the prevailing pursuit of domain-invariant features remain task-sufficient when the observation operator itself changes?
\end{tcolorbox}

\begin{figure}[ht]
  \centering
  \includegraphics[width=0.95\linewidth]{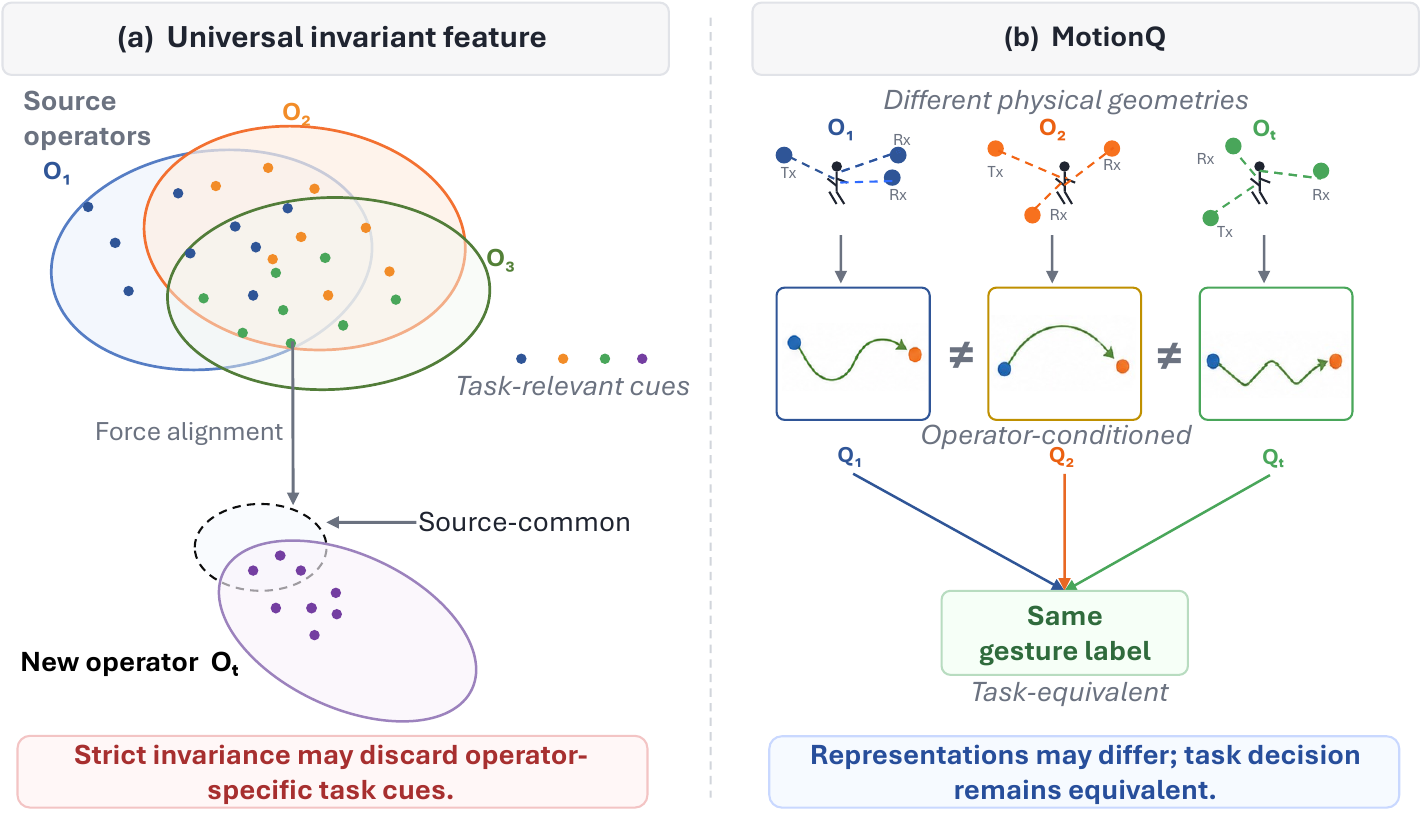}
  \vspace{-0.1in}
  \caption{From strict universal invariance to operator-conditioned task
equivalence.}
  \label{fig:infig}
  \vspace{-0.1in}
\end{figure}


Our answer is: only when task-relevant motion information is jointly observable across the relevant WiFi operators. Different geometries expose different components of the same
motion, so a cue useful under one operator may be absent under another. Prior studies in representation learning have shown that invariance can discard predictive information under non-invertible mappings, distribution shifts, or when task-relevant cues are not confined to shared representations
~\cite{johansson2019support,zhao2019learning,liang2023factorized}. We derive a WiFi-specific common task-observability condition. When this condition fails, alignment cannot recover the missing information, and additional heterogeneous source operators may further shrink the common task-relevant subspace. This does not imply that invariant learning fails in general; rather, it identifies a WiFi sensing regime where invariance is constrained by physical observability. Details please refer to Sec.~\ref{sec:task_observability}.
This raises a natural question: 
\begin{tcolorbox}[colback=gray!25!white, size=title, breakable, boxsep=1mm, colframe=white, before={\vskip1mm}, after={\vskip0mm}]
\textbf{Q2:} If a universal invariant representation is not always appropriate, what should WiFi gesture recognition learn for cross-observation generalization?
\end{tcolorbox}


Our answer is to allow representations to depend on the observation operator while constraining them only through the recognition task, as illustrated in Fig.~\ref{fig:infig}. Guided by this principle, MotionQ follows three stages. First, it encodes the CSI observations of each link, infers several latent user-geometry hypotheses, and uses the known Tx/Rx coordinates to analytically determine how each geometry projects motion onto each link. The resulting geometry-conditioned evidence is fused and decoded into a minimal time-varying two-support motion measure for each geometry hypothesis. Second, the same pipeline processes both the complete observation and every single-link observation. Rather than aligning their representations, MotionQ requires each observation to predict the same gesture label, with a smooth worst-suboperator objective. Third, because the two supports are unlabeled, their arbitrary indices should not become an implicit semantic channel for classification. MotionQ therefore classifies on a \emph{motion quotient} that identifies support-swapped  parameterizations and represents it using permutation-invariant central moments. Each geometry hypothesis then produces a gesture prediction, and the resulting class probabilities are averaged uniformly across hypotheses.

Extensive evaluations demonstrate MotionQ's effectiveness under challenging unseen layouts, extrapolative orientations, and unseen environments. Across six extrapolation tasks, MotionQ averages $89.2\%$, outperforming WiGRUNT~\cite{gu2022wigrunt}, UniFi~\cite{liu2024unifi}, and GesFi~\cite{zhang2026beyond} by $20.1$, $11.0$, and $16.9$ , respectively. In a more realistic setting, where the number of currently available devices may vary, MotionQ ranks first in all three tested configurations and exceeds the strongest WiFi-specific baseline by $12.0$--$18.4$ percentage points. Interestingly, without trajectory supervision, its learned motion supports exhibit consistent gesture-dependent patterns across samples. These emergent structures further suggest that MotionQ captures meaningful motion organization rather than arbitrary discriminative features; we analyze them in Sec.~\ref{sec:support_interpretability}. Our contributions are as follows:
\begin{itemize}[leftmargin=*]
	\item We formulate changes in orientation, links, and transceiver layouts as changes in the wireless observation operator, and establish a common task-observability condition that characterizes when task-sufficient invariant representations are feasible.
	\item  We develop MotionQ, which conditions a motion measure on analytic bistatic geometry, removes arbitrary support identities through a motion quotient, and encourages cross-observation task sufficiency through single-link-retention interventions without representation matching.
	\item MotionQ substantially outperforms state-of-the-art WiFi sensing methods under unseen layouts, extrapolative orientations, and environments, while remaining highly robust when the number and combination of available sensing devices change at deployment.
\end{itemize}

\section{Preliminary}

\subsection{WiFi Gesture Sensing as an Observation Process}
\label{sec:wifi_observation}

Channel State Information (CSI) characterizes the wireless propagation channel between a transmitter and a receiver~\cite{he2026beamforming,han2026rayloc,hu2023muse}, capturing the attenuation and phase changes caused by multipath propagation such as reflection, diffraction, and scattering~\cite{zhang2025wiopen,yan2026diffloc+,tan2025wimap}. For a transmitted signal \(X(f,t)\), the received signal can be written as~\cite{zhang2023wital,chang2026wirainbow,zhang2023wifi}:
\begin{equation}
 Y(f,t)=H(f,t)X(f,t)+N(f,t), 
\end{equation}
where \(Y(f,t)\) and \(X(f,t)\) denote the received and transmitted signals at frequency \(f\) and time \(t\), respectively, \(H(f,t)\) is the CSI describing the wireless channel response, and \(N(f,t)\) denotes measurement noise~\cite{zhao2025baton,zhang2025camlopa}.

Human motion perturbs a subset of propagation paths and thereby introduces time-varying components into CSI. For the \(n\)-th Tx--Rx link, we conceptually decompose its CSI as
\begin{equation}
H_n(f,t) = H_{n,s}(f,t) + H_{n,d}(f,t) + \epsilon_n(f,t), 
\end{equation}
where \(H_{n,s}\) represents the relatively static propagation component, \(H_{n,d}\) captures the dynamic channel variations induced by the executed gesture, and \(\epsilon_n\) summarizes noise and unmodeled propagation effects. Gesture sensing therefore primarily relies on extracting the motion information embedded in \(H_{n,d}\), rather than directly observing the physical motion itself~\cite{li2025muceiver,he2025beam,meng2025metatrack}.

Consider a moving scattering component at position \(r(t)\) with velocity \(v(t)=\dot r(t)\). Let \(T\) and \(R_n\) denote the positions of the transmitter and the receiver of link \(n\), respectively. Its bistatic propagation path length is:
\begin{equation}
d_n(t) = \|r(t)-T\| + \|r(t)-R_n\|.
\end{equation}
Differentiating the path length with respect to time gives:
\begin{equation}
\dot d_n(t) = \left( \frac{r(t)-T}{\|r(t)-T\|} + \frac{r(t)-R_n}{\|r(t)-R_n\|} \right)^{\!T} v(t).
\end{equation}
Defining the bistatic observation vector:
\begin{equation}
q_n(r) = -\frac{1}{\lambda} \left( \frac{r-T}{\|r-T\|} + \frac{r-R_n}{\|r-R_n\|} \right), 
\end{equation}
where \(\lambda\) is the carrier wavelength, the corresponding Doppler shift satisfies
\begin{equation}
f_{D,n}=q_n(r)^Tv.
\end{equation}

This relation reveals a fundamental property of WiFi gesture sensing: a WiFi link does not observe complete motion; it observes a geometry-dependent projection of that motion. Changing the user orientation, available links, or transceiver placement changes the corresponding bistatic observation vectors and thus changes how the same physical gesture appears in wireless measurements. In this work, we focus specifically on arm- and hand-based gestures, as commonly considered in WiFi gesture sensing~\cite{zhang2021widar3,li2020wihf,feng2025imbalanced}, rather than general full-body activity recognition.

\subsection{Cross-Observation Generalization}
\label{sec:cross_observation}

WiFi sensing is affected by different types of deployment variations. Some variations primarily perturb the statistical appearance of wireless measurements. Examples include hardware-dependent distortions and environment changes. In contrast, changes in user orientation or transceiver placement alter the geometric relationship through which human motion is observed. Such changes do not merely modify the appearance of the signal; they change the physical mapping from the same gesture motion to its wireless measurement. However, most current methods often treat such variations as style factors and attempt to suppress them through domain-invariant representations.

We formalize the latter case through a wireless observation operator. Let \(m\in\mathcal{M}\) denote a gesture motion, where \(\mathcal{M}\) is the space of possible gesture executions, and let \(y(m)\in\mathcal{Y}\) denote its gesture label. We define an observation operator \(g\in\mathcal{G}\) by the physical geometry through which WiFi observes the gesture,
\begin{equation}
g=\left(T,\mathbf{R},p,\theta\right),
\end{equation}
where \(T\) denotes the transmitter position, \(\mathbf{R}=\{R_1,\ldots,R_N\}\) denotes the receiver positions, \(p\) denotes the user position, and \(\theta\) denotes the user orientation. Given \(g\), the wireless sensing process is represented by
\begin{equation}
F_g:\mathcal{M}\rightarrow\mathcal{X}, \qquad x_g=F_g(m),
\end{equation}
where \(F_g\) maps the physical gesture \(m\) to the corresponding wireless observation \(x_g\), such as its CSI-derived temporal and Doppler measurements.

Because \(g\) determines the bistatic projection geometry derived in Sec.~2.1, the same gesture is generally expected to produce different observations under different operators:
\begin{equation}
F_{g_1}(m)\neq F_{g_2}(m).
\end{equation}
This difference is therefore not necessarily a domain artifact to be removed; it can be a physically correct consequence of observing the same motion from different geometries.

\subsection{Local Common Task Observability under WiFi Operators}
\label{sec:task_observability}

Section~\ref{sec:cross_observation} shows that different observation operators physically expose different aspects of the same gesture. This raises a question for invariant learning in WiFi sensing: when can different operators produce a common representation without losing information required for gesture recognition? We answer it through a common task-observability condition.

\textbf{Local WiFi observation model.}
Let $s\in\mathbb{R}^{d}$ denote the local state of a gesture around a particular motion state. Here, $s$ is not restricted to a 2-D velocity vector; it may contain multiple motion components or other local variables that characterize the gestures. For an observation operator $g$, the wireless measurement can be locally approximated by
\begin{equation}
    x_g=A_g s,
    \label{eq:local_observation}
\end{equation}
where $A_g\in\mathbb{R}^{n_g\times d}$ is the local observation Jacobian induced by $g$. For the single moving component discussed in Sec.~\ref{sec:wifi_observation}, this model reduces to the bistatic projection $f_D=q_g^{T}v$. More generally, $A_g$ describes how the complete local gesture state is mapped to the CSI-derived observation under a particular user orientation and transceiver geometry.

We further represent the task-relevant local gesture information by a linear task variable
\begin{equation}
    u=Bs,
    \label{eq:task_coordinate}
\end{equation}
where $B\in\mathbb{R}^{d_u\times d}$ maps the local gesture state $s$ to the variations that must be retained for gesture discrimination. The final gesture classifier may remain nonlinear; $u$ only specifies the task-relevant information that must be preserved by the representation.


Consider now a \emph{strict operator-invariant representation}
\begin{equation}
    z=Cs,
\end{equation}
where the same $C\in\mathbb{R}^{d_z\times d}$ is used to describe the gesture regardless of the observation operator. Such a representation must satisfy two requirements. First, $z$ must be recoverable from every wireless observation, i.e., for each $g$ there exists $D_g$ such that
\begin{equation}
    z=D_gx_g,
    \quad\text{and hence}\quad
    C=D_gA_g.
    \label{eq:recoverable}
\end{equation}
Second, $z$ must remain sufficient for the recognition task, i.e., there exists $L$ such that
\begin{equation}
    u=Wz,
    \qquad
    B=WC.
\label{eq:sufficient}
\end{equation}

These two requirements lead directly to the following condition:
\begin{theorem}[Local Common Task Observability]
\label{thm:common_observability}
Under the local linear WiFi observation model in Eq.~(10), there exists a
linear common representation $z=Cs$ that is
(i) operator-wise recoverable, i.e., for every $g$ there exists $D_g$ such
that $z=D_gx_g$, and
(ii)  task-sufficient, i.e., there exists $W$ such that $u=Wz$,
if and only if
\begin{equation}
    \operatorname{Row}(B)
    \subseteq
    \bigcap_{g\in\mathcal G}\operatorname{Row}(A_g).
\label{eq:common_observability}
\end{equation}
\end{theorem}

\textbf{Scope.}
This is a local statement about exact linear common recoverability under geometry-induced WiFi operators, rather than a general impossibility theorem for arbitrary nonlinear or distribution-level domain generalization.  Allowing $D_g$ to depend on the operator makes the condition permissive: if it fails, even an operator-aware recovery map cannot
recover the task-relevant variations lost by $A_g$ from $x_g$.
Practical alignment methods need not  enforce exact equality, but they cannot create task information lying in $\operatorname{Null}(A_g)$.

\begin{proof}
If a task-sufficient common representation exists, Eq.~\eqref{eq:sufficient} Choosing $C=B$ and $W=I$, and therefore
\[
    \operatorname{Row}(B)
    \subseteq
    \operatorname{Row}(C).
\]
Meanwhile, Eq.~\eqref{eq:recoverable} gives $C=D_gA_g$ for every $g$, yielding
\[
    \operatorname{Row}(C)
    \subseteq
    \operatorname{Row}(A_g),
    \qquad \forall g.
\]
Combining the two relations gives Eq.~\eqref{eq:common_observability}.

Conversely, if Eq.~\eqref{eq:common_observability} holds, then for every $g$ there exists a matrix $D_g$ such that $B=D_gA_g$. Choosing $C=B$ and $L=I$ gives a representation $z=Bs$ that can be recovered from every observation and is itself task-sufficient. Therefore, the condition is also sufficient.
\end{proof}

Theorem~\ref{thm:common_observability} gives a direct interpretation of the assumption behind operator-invariant learning. Define
\[
    \mathcal{O}_g=\operatorname{Row}(A_g)
\]
as the observable gesture subspace of operator $g$, i.e., the local gesture variations that can be recovered from its wireless observation, and
\[
    \mathcal{T}=\operatorname{Row}(B)
\]
as the task-relevant gesture subspace. Importantly, $\mathcal{O}_g$ is not merely a spatial direction. It may contain any locally observable gesture variation, including combinations of motion components, temporal motion coefficients, and scattering-induced changes. The bistatic direction $q_g$ in Sec.~\ref{sec:wifi_observation} is only its simplest special case.

The theorem states that a strict linear common representation is possible only if
\begin{equation}
    \mathcal{T}
    \subseteq
    \mathcal{O}_{g_1}
    \cap\mathcal{O}_{g_2}
    \cap\cdots,
\end{equation}
i.e., \emph{every gesture variation required by the task must be observable under every operator from which the same representation is expected to be recovered}.

An equivalent form of Eq.~\eqref{eq:common_observability} is
\begin{equation}
    \operatorname{Null}(A_g)
    \subseteq
    \operatorname{Null}(B),
    \qquad \forall g.
\label{eq:null_condition}
\end{equation}

If $\Delta s\in\operatorname{Null}(A_g)$, then
\[
    A_g(s+\Delta s)=A_gs.
\]
Exact task sufficiency therefore requires $B\Delta s=0$; otherwise, operator $g$ cannot distinguish two states that differ in task-relevant information. Retaining such a variation violates recoverability from $g$, whereas removing it sacrifices discriminability. Stronger alignment cannot reconstruct information absent from $x_g$.


\textbf{Why strict invariance can lose task information.}
Suppose a gesture variation is useful for recognition and is observable under one operator, but is not observable under another. It then belongs to $\mathcal{T}$ but not to the common observable subspace in Eq.~\eqref{eq:common_observability}. A strictly invariant learner faces an unavoidable conflict: retaining this variation improves task discriminability but makes the representation unavailable from every operator, whereas removing it improves cross-operator consistency at the cost of discarding task-relevant information. Therefore, strict cross-operator invariance can create a fundamental tension between \emph{representation consistency} and \emph{task sufficiency}.

Source invariance does not imply target observability.
Let
\[
    \mathcal{S}=\{g_1,\ldots,g_J\}
\]
denote the source operators and define their common observable subspace as
\begin{equation}
    \mathcal{I}_{\mathcal S}
    =
    \bigcap_{g\in\mathcal S}
    \mathcal{O}_g.
\label{eq:source_intersection}
\end{equation}
An invariant representation learned from the source operators can only rely on information contained in $\mathcal{I}_{\mathcal S}$. However, for the same representation to be recovered from an unseen target operator $g_t$, it must additionally satisfy
\begin{equation}
    \operatorname{Row}(C)
    \subseteq
    \mathcal{O}_{g_t}.
\label{eq:target_condition}
\end{equation}
Nothing in source-domain feature alignment guarantees Eq.~\eqref{eq:target_condition}. When the target orientation or transceiver layout introduces an unseen observation geometry, some task cues used by the source-common representation may no longer be observable in the same way. Hence, source invariance does not imply target observability, explaining why performance under observation-operator extrapolation cannot be guaranteed by domain invariance alone.

\textbf{Adding heterogeneous source operators.}
The same analysis also reveals a counter-intuitive property of invariant learning. For a source set $\mathcal S$, the task-relevant information compatible with
strict common recoverability is:
\begin{equation}
    \mathcal T_{\mathcal S}
    =
    \mathcal T
    \cap
    \mathcal{I}_{\mathcal S},
    \qquad
    \mathcal T_{\mathcal S\cup\{g'\}}
    \subseteq
    \mathcal T_{\mathcal S}.
\end{equation}

Therefore, under strict invariance, the task-relevant information that remains jointly usable across all source operators can only stay unchanged or decrease as additional operators are introduced. If the new operator is geometrically similar to the existing ones, the reduction may be negligible. In contrast, a sufficiently different operator can remove task-relevant variations from the common subspace, increasing the conflict between invariance and discriminability. Adding heterogeneous source domains may therefore cause, rather than necessarily prevent, negative transfer.

\textbf{Simultaneous link aggregation versus cross-operator commonality.}
It is important to distinguish adding a \emph{simultaneous wireless observation} from adding an \emph{invariance constraint across operators}. Suppose one gesture is observed simultaneously by $N$ links under the same observation configuration. Let
\[
    x_{g,n}=A_{g,n}s
\]
be the observation of link $n$. Define the joint observation operator and the corresponding observation as:
\begin{equation}
    A_g^{\mathrm{joint}}
    :=
    \begin{bmatrix}
        A_{g,1}\\
        \vdots\\
        A_{g,N}
    \end{bmatrix},
    \qquad
    x_g=A_g^{\mathrm{joint}}s .
\end{equation}
Its observable subspace satisfies
\begin{equation}
    \operatorname{Row}(A_g^{\mathrm{joint}})
    =
    \sum_{n=1}^{N}\operatorname{Row}(A_{g,n}),
\end{equation}
where the sum denotes the vector-space sum and need not be a direct sum. Hence, adding a complementary link can only preserve or enlarge the jointly observable gesture subspace.


\begin{figure*}[ht]
  \centering
  \includegraphics[width=0.76\linewidth]{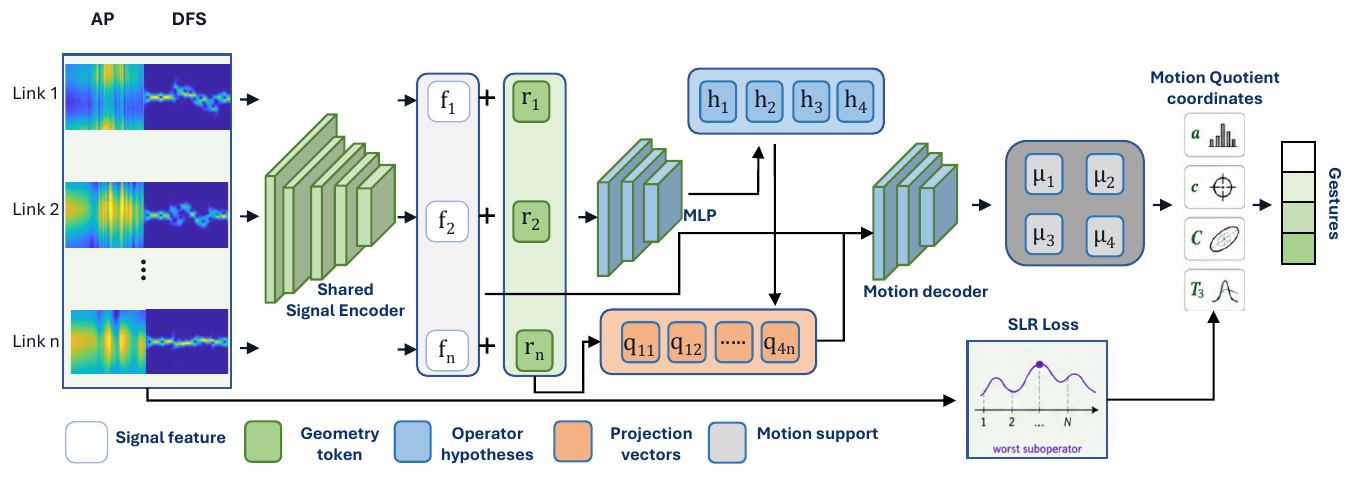}
  \vspace{-0.1in}
  \caption{System Overview.}
  \label{fig:overview}
  \vspace{-0.15in}
\end{figure*}

By contrast, requiring the same representation to be recoverable across heterogeneous operators restricts it to
\begin{equation}
    \bigcap_{g\in\mathcal S}
    \operatorname{Row}(A_g),
\label{eq:operator_intersection}
\end{equation}
which can only preserve or shrink as more operators are included. We refer to this contrast as the simultaneous link aggregation versus cross-operator commonality: \emph{simultaneous observations add observable information, whereas strict invariance across heterogeneous operators intersects it}.

The above analysis yields two key implications for cross-observation WiFi sensing. First, forcing different observation operators to share a numerically identical representation is unnecessary when useful task information is operator-dependent. Second, abandoning invariance altogether would allow a model to memorize arbitrary operator-specific shortcuts. A desirable representation should therefore be \emph{conditioned on how the gesture is observed}, while being constrained to retain only information sufficient for the recognition task. 
Together, these results motivate representations that are conditioned on the observation operator but constrained by task sufficiency rather than numerical equality across operators. This principle leads to MotionQ.

\section{System Overview}
\label{sec:overview}

Figure~\ref{fig:overview} summarizes MotionQ in three stages. First, shared AP/DFS encoders extract per-link evidence, while latent geometry hypotheses convert the known Tx/Rx coordinates into bistatic operators that multiplicatively condition link fusion and two-support motion-measure generation. Second, the same pipeline processes the complete observation and every single-link observation, and a smooth worst-suboperator objective encourages each view to retain information sufficient for the gesture task without matching representations. Third, permutation-invariant central moments provide motion-quotient coordinates, and the conditional class probabilities of all geometry hypotheses are averaged uniformly.




\section{Method}


\subsection{Motion Measure Generation}
\label{sec:motion_measure_generation}

MotionQ first constructs an operator-conditioned two support motion measure that serves as the basis of the motion quotient. Bistatic propagation geometry motivates the operator conditioning, while coordinated limb motion motivates a compact two-support latent representation.

\textbf{Per-link wireless evidence.}
For each wireless link $i$, MotionQ takes its amplitude--phase (AP) representation $A_i$ and Doppler frequency spectrum (DFS) $D_i$ as complementary observations of gesture-induced channel dynamics. AP and DFS are processed by two modality-specific signal encoders whose parameters are shared across all links:
\begin{equation}
    e_i^{A}=E_A(A_i), \qquad
    e_i^{D}=E_D(D_i).
\end{equation}
Their outputs are concatenated and projected into a common per-link feature,
\begin{equation}
    f_i
    =
    P\!\left([e_i^{A},e_i^{D}]\right)
    \in\mathbb{R}^{d_f},
    \label{eq:pair_feature}
\end{equation}
where $d_f=512$ in our implementation. Sharing the same encoders across links prevents the network architecture from depending on a particular receiver identity or link count.


\textbf{Latent observation-operator hypotheses.}
The physical interpretation of $f_i$ depends on the geometric relationship among the transmitter, receiver, and user, while the exact user position and orientation are generally unavailable at deployment. Rather than committing to a single geometric explanation, MotionQ maintains a small set of latent observation-operator hypotheses,
\begin{equation}
    \mathcal{H}
    =
    \{\xi_1,\ldots,\xi_H\},
    \qquad
    \xi_h=(p,o_h),
\end{equation}
where $p\in\mathbb{R}^2$ is a candidate user position shared by all hypotheses and $o_h\in\mathbb{R}^2$ is the unit orientation of hypothesis $h$.

The transceiver coordinates required by MotionQ can be obtained with lightweight deployment-time calibration; for example, PerceptAlign~\cite{jia2026breaking} demonstrates a simple coordinate-registration procedure. Given receiver coordinate $R_i$, each available link is encoded as
\begin{equation}
    \zeta_i
    =
    E_{\mathrm{set}}
    \left(
        [f_i,\psi_{\mathcal V}(R_i)]
    \right),
\end{equation}
where $\mathcal V$ denotes the currently active link set and $\psi_{\mathcal V}(\cdot)$ is its layout-relative coordinate encoding. Its exact construction is provided in Appendix~\ref{app:modetail}. 

The active links are then aggregated by
\begin{equation}
    \bar\zeta
    =
    \frac{1}{|\mathcal V|}
    \sum_{i\in\mathcal V}\zeta_i.
\end{equation}

A shared operator network maps $\bar\zeta$ to the position $p$ and the $H$ orientations $o_h$. The orientations are initialized from evenly spaced anchors, while their residuals and the position remain learnable. The exact
parameterization is provided in Appendix~\ref{app:modetail}. Each $\xi_h$ is a sample-level geometric explanation of the active observation rather than a supervised estimate of the true user pose.

\textbf{Bistatic operator conditioning.}
For each hypothesis $h$ and wireless link $i$, the candidate user
geometry together with the known Tx and Rx coordinates analytically
determines the bistatic observation vector. Following Sec.~\ref{sec:wifi_observation}, we
compute
\begin{equation}
    q_{hi}^{\mathrm{w}}
    =
    -\frac{1}{\lambda}
    \left(
               \frac{p-T}{\lVert p-T\rVert_2}
        +
        \frac{p-R_i}{\lVert p-R_i\rVert_2}
    \right),
\end{equation}
where $T$ and $R_i$ denote the Tx and Rx positions and $\lambda$ is the
carrier wavelength. The vector is then expressed in the candidate body
coordinate system:
\begin{equation}
    q_{hi}
    =
    \operatorname{Rot}(o_h)^\top q_{hi}^{w},
\end{equation}
where $\operatorname{Rot}(o_h)$ is the two-dimensional rotation induced by orientation $o_h$.

MotionQ injects this analytically derived observation operator before
motion-support generation rather than treating geometry as metadata
appended to the classifier. Let
$\tilde q_{hi}=\lambda q_{hi}/2$ denote the dimensionless bistatic coordinate.The link feature is conditioned as
\begin{equation}
    \widetilde f_{hi}
    =
    \sum_{k=1}^{4}
    \phi_k(\widetilde q_{hi})E_k(f_i),
\end{equation}
where $\phi$ is a four-term directional basis and $E_k$ are independent bias-free signal mappings, as detailed in Appendix~\ref{app:modetail}. The conditioned features are fused for each hypothesis:
\begin{equation}
    \widetilde f_h
    =
    \frac{1}{|\mathcal V|}
    \sum_{i\in\mathcal V}\widetilde f_{hi}.
\end{equation}
The basis contains no constant component, and no coordinate-independent residual is provided from $f_i$ to $\widetilde f_{hi}$. Therefore, link evidence reaches motion-measure generation only after being modulated by the analytically derived coordinate.

Human gestures generally contain multiple simultaneous effective motion components whose Doppler responses are superimposed at the receiver \cite{zhang2021widar3,abdulatif2017real}. We represent their instantaneous velocities and relative contributions by the discrete measure:
\begin{equation}
    \nu_t
    =
    \sum_{k=1}^{K_t}
    w_{k,t}\delta_{v_{k,t}},
    \qquad
    w_{k,t}\geq 0,
    \quad
    \sum_{k=1}^{K_t}w_{k,t}=1,
    \label{eq:latent_velocity_measure}
\end{equation}
where $v_{k,t}\in\mathbb{R}^{2}$ is an effective motion component and $w_{k,t}$ is its relative contribution. These components are not assigned predefined body-part identities.

Under hypothesis $h$, link $i$ observes each velocity only through the scalar bistatic projection $\pi_{hi}(v)=q_{hi}^{\top}v$. Applying this projection to every
component of $\nu_t$ yields the one-dimensional Doppler measure
\begin{equation}
    \eta_{hi,t}
    =
    (\pi_{hi})_{\#}\nu_t
    =
    \sum_{k=1}^{K_t}
    w_{k,t}\delta_{q_{hi}^{\top}v_{k,t}}.
    \label{eq:bistatic_measure_pushforward}
\end{equation}
where $(\pi_{hi})_{\#}$ maps the support velocities through $\ell_{hi}$ while preserving their weights. Thus, different links provide different one-dimensional projections of the same underlying velocity measure, motivating MotionQ to infer a compact velocity measure from the geometry-conditioned multi-link evidence.

\textbf{Minimal two-support motion measure.}
The remaining question is how complex the latent velocity measure should be. Human upper-limb motion is highly coordinated rather than an arbitrary high-dimensional process. Biomechanical studies show that a small number of kinematic synergies explain most of the variation in multi-joint reaching, with low-order components capturing dominant motion patterns~\cite{tang2019kinematic}. Motor-control studies likewise describe complex reaching movements as compositions of elementary submovements~\cite{rohrer2004submovements}. These findings motivate a compact, component-based description of gesture motion, but do not prescribe the number of components. MotionQ instead determines the minimum support required to represent the nontrivial motion structure exposed by the wireless observation model.

A single-support measure
\begin{equation}
    \nu_t^{(1)}
    =
    \delta_{\bar v_t}
    \label{eq:one_support_measure}
\end{equation}
contains only an aggregate velocity and has zero central moments
of order two and above,
\begin{equation}
    \Sigma_t=0,
    \qquad
    M_t^{(j)}=0,
    \quad j\geq 3.
\end{equation}
It therefore cannot represent differential motion among simultaneously contributing velocity components. Conversely, adding many unconstrained supports introduces degrees of freedom that finite wireless projections need not identify. MotionQ consequently adopts the smallest non-degenerate extension of Eq.~\eqref{eq:one_support_measure}: a two-support velocity measure.

For each operator hypothesis $h$, a shared decoder maps $\widetilde f_h$ to $N_t$ motion states ($N_t=28$ in our implementation). At time $t$, it predicts an activity $a_{h,t}$, mean velocity $\bar v_{h,t}\in\mathbb{R}^2$, relative split $\Delta v_{h,t}\in\mathbb{R}^2$, and balance $\rho_{h,t}$. They define
\begin{align}
    v_{1,h,t}
    &=
    \bar v_{h,t}
    +(1-\rho_{h,t})\Delta v_{h,t},
    \label{eq:support_one}
    \\
    v_{2,h,t}
    &=
    \bar v_{h,t}
    -\rho_{h,t}\Delta v_{h,t}.
    \label{eq:support_two}
\end{align}
and the normalized moving measure
\begin{equation}
    \bar{\mu}_{h,t}
    =
    \rho_{h,t}\delta_{v_{1,h,t}}
    +
    (1-\rho_{h,t})\delta_{v_{2,h,t}},
    \label{eq:two_support_measure}
\end{equation}
with activity-weighted measure
\begin{equation}
    \mu_{h,t}
    =
    a_{h,t}\bar{\mu}_{h,t}.
    \label{eq:activity_motion_measure}
\end{equation}
The parameterization explicitly separates aggregate and relative motion:
\begin{equation}
    \mathbb{E}_{\bar\mu_{h,t}}[v]
    =
    \bar v_{h,t},
    \qquad
    v_{1,h,t}-v_{2,h,t}
    =
    \Delta v_{h,t}.
    \label{eq:carrier_split_identity}
\end{equation}

We constrain $a_{h,t}\geq 0$ and $\rho_{h,t}\in(0,1)$ using the
corresponding output activations. More importantly, two supports introduce the first non-zero internal moments beyond the mean. Their covariance and third central moment are:
\begin{align}
    \Sigma_{h,t}
    &=
    \rho_{h,t}(1-\rho_{h,t})
    \Delta v_{h,t}\Delta v_{h,t}^{\top},
    \label{eq:two_support_covariance}
    \\
    T_{3,h,t}
    &=
    M^{(3)}_{h,t}
    =
    \rho_{h,t}(1-\rho_{h,t})(1-2\rho_{h,t})
    \Delta v_{h,t}^{\otimes 3}.
    \label{eq:two_support_third_moment}
\end{align}
Here, $\Sigma_{h,t}$ captures a single dominant axis and magnitude of velocity dispersion, while $M^{(3)}_{h,t}$ distinguishes balanced from unequal contributions of the two components. Hence, two supports are the minimum finite structure that augments aggregate motion with directional dispersion and component imbalance, without introducing an unconstrained field. The two supports have no predefined semantic identities or correspondence to particular body parts. 




\subsection{Task-Equivalent Training}
\label{sec:task_equivalent_training}

Operator-conditioned motion representations are allowed to vary with the observation geometry. The remaining challenge is to prevent them from relying on task cues that are useful only under a particular observation operator. Directly aligning representations across observations is undesirable, since it would reintroduce the invariance constraint discussed in Sec.~\ref{sec:task_observability}. MotionQ therefore enforces equivalence only at the task level.

\textbf{Single-link-retention operator intervention.}
Let $\mathcal V$ denote the set of synchronized links in a training sample. For each link $i\in\mathcal V$, we define the single-link view
\begin{equation}
    \mathcal V_i=\{i\},
    \label{eq:single_link_observation}
\end{equation}
which retains only link $i$. MotionQ processes both the complete observation $x_{\mathcal V}$ and every single-link observation $x_{\mathcal V_i}$ using the same model. Each single-link observation is forwarded through the complete pipeline with $\mathcal V_i$ as its active link set, including operator-hypothesis inference, bistatic conditioning, and motion-measure generation.

Crucially, MotionQ imposes no constraint of the form $\lVert Q_{\mathcal{L}}-Q_{\mathcal{S}_i}\rVert$. The quotients inferred from different observations may differ; each is only required to remain sufficient for the same gesture label:
\begin{equation}
    Q_{\mathcal V}
    \neq
    Q_{\mathcal V_i}
    \text{ is allowed},
    \qquad
    \hat y_{\mathcal V}
    =
    \hat y_{\mathcal V_i}
    =
    y.
    \label{eq:single_link_task_equivalence}
\end{equation}
This realizes the principle in Sec.~\ref{sec:task_observability}: MotionQ promotes task sufficiency under individual-link suboperators and provides a local robustness bias for broader operator changes without requiring representation invariance.

For the $i$-th single-link intervention, let
\begin{equation}
     \ell_i
    =
    -\log p\left(y\mid x_{\mathcal V_i}\right).
    \label{eq:single_link_risk}
\end{equation}
be its classification loss. We aggregate the single-link risks using a smooth worst-suboperator objective~\cite{li2021tilted},
\begin{equation}
    \mathcal L_{\mathrm{SLR}}
    =
    \tau
    \left[
        \log
        \sum_{i\in\mathcal V}
        \exp\left(\frac{\ell_i}{\tau}\right)
        -
        \log|\mathcal V|
    \right],
    \label{eq:slr_loss}
\end{equation}
where $\tau=0.2$. As $\tau$ decreases, Eq.~\eqref{eq:slr_loss} approaches the worst single-link intervention while remaining differentiable. We refer to it as the single-link-retention task-equivalence loss.

\subsection{Motion Quotient-Based Classification}
\label{sec:quotient_classification}

The two-support motion measure generated in Sec.~\ref{sec:motion_measure_generation}
admits two equivalent ordered parameterizations because the two supports have
no predefined identities~\cite{monteiller2019alleviating}. We define a \emph{motion quotient} $Q_t$ as the equivalence class obtained by identifying parameterizations that differ only by exchanging the two supports, so that arbitrary support indices cannot acquire unintended semantic meaning during learning.

\textbf{Motion quotient.}
For operator hypothesis $h$ and time step $t$, suppressing the
hypothesis index below for clarity, let
\begin{equation}
    \vartheta_t
    =
    (a_t,\rho_t,v_{1,t},v_{2,t})
\end{equation}
denote an ordered parameterization of the motion measure. Swapping the
two supports gives
\begin{equation}
    \sigma(\vartheta_t)
    =
    (a_t,1-\rho_t,v_{2,t},v_{1,t}),
\end{equation}
which represents exactly the same measure. We therefore identify these
two parameterizations and define
\begin{equation}
    Q_t
    =
    [\vartheta_t]_{S_2},
    \qquad
    \mathcal{Q}_2
    =
    \Theta_2/S_2,
\label{eq:motion_quotient}
\end{equation}
where $\Theta_2$ denotes the ordered two-support parameter space and $S_2$ is the permutation group of two elements.


\textbf{Quotient coordinates and classification.}
The mean velocity $\bar v_t$, covariance $\Sigma_t$, and third central moment $M_t^{(3)}$ derived in Sec.~\ref{sec:motion_measure_generation} are invariant to the support permutation and therefore define coordinates on the quotient. MotionQ represents each time step as:
\begin{equation}
    \Phi(Q_t)
    =
    \left[
        a_t,\,
        a_t\bar v_t,\,
        a_t\operatorname{vec}_s(\Sigma_t),\,
        a_t\operatorname{vec}_s(M_t^{(3)})
    \right]
    \in\mathbb{R}^{10}.
    \label{eq:quotient_coordinates}
\end{equation}
where $\operatorname{vec}_{s}(\cdot)$ retains the independent entries of a symmetric tensor. As shown by the two-support construction in Sec.~\ref{sec:motion_measure_generation}, these moments generically determine the underlying measure up to exactly the support permutation removed by Eq.~\eqref{eq:motion_quotient}. Thus, the classifier discards the arbitrary support identity without introducing an additional unconstrained semantic representation.

For operator hypothesis $h$, the quotient sequence
\begin{equation}
    Q_h
    =
    \{\Phi(Q_{h,t})\}_{t=1}^{N_t}.
\end{equation}
is passed to the shared classifier to obtain the conditional class probability $p(y\mid x,h)$. MotionQ assigns a uniform prior to the latent observation-operator hypotheses and marginalizes their predictions as
\begin{equation}
    p(y\mid x)
    =
    \frac{1}{H}
    \sum_{h=1}^{H}
    p(y\mid x,h).
\end{equation}
Importantly, marginalization is performed over conditional class probabilities rather than by averaging logits or motion measures generated under different geometric hypotheses. Combining the classification objective and the single-link-retention task-equivalence loss gives
\begin{equation}
    \mathcal{L}
    =
    -\log p(y\mid x_{\mathcal L})
    +
    \lambda_{\mathrm{SLR}}\mathcal{L}_{\mathrm{SLR}}.
\end{equation}

\section{Evaluations}

\begin{figure}[ht]
  \centering
  \vspace{-0.1in}
  \includegraphics[width=0.78\linewidth]{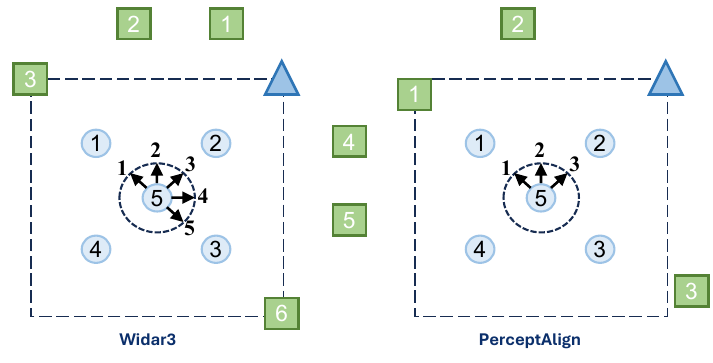}
   \vspace{-0.1in}
  \caption{Example transceiver layouts in Widar3.0 and PerceptAlign, with four layouts used in PerceptAlign.}
  \label{fig:wplayouts}
  \vspace{-0.2in}
\end{figure}

\subsection{Experimental Setup}
\label{sec:experimental_setup}


\begin{table*}[t]
\centering
\caption{Overall recognition accuracy (\%). W4-D, P, and -E denote
cross-direction, position, and environment evaluation,
respectively. W7 $m\!\rightarrow\!n$ denotes training with $m$ links and
deployment with $n$ available links.}
\label{tab:overall_performance}
\small
\setlength{\tabcolsep}{7.0pt}
\renewcommand{\arraystretch}{1.08}
\begin{tabular}{@{}lrrrrrr@{}}
\toprule
\textbf{Protocol}
& \textbf{MotionQ}
& \textbf{WiGRUNT~\cite{gu2022wigrunt}}
& \textbf{UniFi~\cite{liu2024unifi}}
& \textbf{GesFi~\cite{zhang2026beyond}}
& \textbf{CORAL~\cite{sun2016deep}}
& \textbf{DANN~\cite{ganin2016domain}} \\
\midrule

\multicolumn{7}{@{}l}{\textit{Multi-factor cross-observation generalization}} \\
W1
& \textbf{93.699}
& 80.860
& 88.648
& 80.756
& \underline{92.147}
& 91.825 \\
W2
& \textbf{94.590}
& 86.465
& 87.380
& 89.300
& \underline{92.103}
& 89.827 \\
W3
& \textbf{94.080}
& 80.447
& 89.288
& 84.158
& \underline{92.490}
& 92.422 \\
\addlinespace[1pt]

\multicolumn{7}{@{}l}{\textit{Sufficiently covered observations}} \\
W4-D
& \underline{99.130}
& 91.920
& \textbf{99.400}
& 96.450
& 98.177
& 98.176 \\
W4-P
& \textbf{99.478}
& 92.070
& 99.180
& 99.260
& \underline{99.319}
& 99.138 \\
W4-E
& 97.024
& 85.190
& 97.730
& \textbf{99.320}
& \underline{98.380}
& 97.718 \\
\addlinespace[1pt]

\multicolumn{7}{@{}l}{\textit{Source-operator coverage and extrapolation}} \\
W5
& \textbf{74.859}
& 64.770
& 68.892
& 62.094
& 71.161
& \underline{71.635} \\
W6
& \textbf{72.088}
& 62.605
& 63.261
& 66.222
& \underline{66.627}
& 65.257 \\
\cmidrule(lr){1-7}
\textbf{W1--W6 Mean}
& \textbf{90.618}
& 80.541
& 86.722
& 84.695
& \underline{88.800}
& 88.250 \\
\addlinespace[2pt]

\multicolumn{7}{@{}l}{\textit{Unavailable links at deployment}} \\
W7 $3\!\rightarrow\!2$
& \textbf{89.943}
& 32.945
& 77.962
& 34.810
& 84.770
& \underline{86.266} \\
W7 $6\!\rightarrow\!2$
& \textbf{91.794}
& 24.325
& 73.396
& 21.220
& 80.770
& \underline{86.386} \\
W7 $6\!\rightarrow\!3$
& \textbf{93.742}
& 25.021
& 80.161
& 21.754
& 85.726
& \underline{89.991} \\
\addlinespace[1pt]

\multicolumn{7}{@{}l}{\textit{Full-body motion stress test on PerceptAlign}} \\
P1
& \underline{79.794}
& 62.022
& 78.326
& 79.327
& \textbf{80.424}
& 79.381 \\
P2
& \underline{78.676}
& 66.403
& 76.937
& 74.893
& \textbf{79.293}
& 77.566 \\
P3
& \underline{73.721}
& 63.834
& 71.429
& 71.668
& \textbf{74.561}
& 72.570 \\
\midrule
\textbf{Overall Mean}
& \textbf{88.044}
& 65.634
& 82.285
& 70.088
& 85.425
& \underline{85.583} \\
\bottomrule
\end{tabular}
\end{table*}

\textbf{Datasets and protocols.}
We evaluate all methods on the Widar3.0~\cite{zhang2021widar3} and the PerceptAlign dataset~\cite{jia2026breaking} (Figure.~\ref{fig:wplayouts}). We use six gestures in Widar3.0: Push--Pull, Sweep, Clap, Slide, Draw-O, and Draw-Zigzag. We construct seven protocols. W1--W3 evaluate generalization under simultaneous changes in users, environments, layouts, and orientations, using unseen two- or three-link layouts. Here, cross-layout means that the receiver combination is unseen during training. For example, W1 is trained on layouts $[2,5]$ and $[2,6]$, and evaluated on the layouts $[1,3]$ and $[1,4]$. W4 retains all six links to evaluate conventional cross-domain recognition under sufficient observation coverage. W5 and W6 vary the geometric coverage of source orientations to study observation-operator coverage. W7 deploys frozen models trained with three or six links over all two- or three-link subsets, to evaluate robustness to unavailable links at deployment. Because no other public multi-link gesture dataset provides the required transceiver coordinates, we additionally use PerceptAlign as stress test. We use four base activities---Stretch, Lunge, Squat, and Jump, and do not treat direction-specific execution variants as different classes. P1--P3 evaluate generalization under simultaneous changes in scenes, receiver layouts, and execution orientations. Unlike the compact limb gestures in Widar3.0, these activities involve more complex whole-body motion and are used to examine the scope, rather than the primary target of MotionQ. More details are provided in Appendix~\ref{app:protocols}.

\textbf{Baselines.}
We compare MotionQ with three representative WiFi-specific systems and two generic domain generalization methods. \textbf{WiGRUNT}~\cite{gu2022wigrunt} learns domain-robust spatial--temporal cues using dual attention; \textbf{UniFi}~\cite{liu2024unifi} learns consistent multi-view representations through cross-view fusion and contrastive regularization; and \textbf{GesFi}~\cite{zhang2026beyond} automatically mines latent source domains before adversarial alignment. We further include two widely used and competitive generic domain-generalization methods, \textbf{CORAL}~\cite{sun2016deep} and \textbf{DANN}~\cite{ganin2016domain}, representing covariance matching and adversarial domain confusion, respectively. Each WiFi-specific baseline uses the input reported to perform best for its original architecture. CORAL and DANN use exactly the same AP/DFS dual-branch shared encoder as MotionQ, and UniFi also adopts this dual-branch architecture, but uses P and DFS as inputs because its ablation study showed that AP performs worse than P alone. Source-domain labels of CORAL and DANN are constructed from all annotated factors that vary in each protocol, such as environment, layout, and orientation. This favorable, protocol-informed setting requires no target data, but assumes advance knowledge of the types of deployment shifts under evaluation. Such shift-aware partitioning is difficult to guarantee in practical WiFi sensing, where multiple coupled factors may be unknown or poorly characterized by physical labels~\cite{zhang2026beyond}. We therefore use CORAL and DANN as optimistic references under idealized shift knowledge.

\textbf{Implementation.}
All methods follow the same 40-epoch training schedule and use
ResNet-18 backbones~\cite{he2016deep}. The learning rate starts at $10^{-4}$ and is reduced to $10^{-5}$ and $10^{-6}$ after epochs 10 and 25, respectively. MotionQ uses $K=2$ motion supports, $H=4$ operator hypotheses, $\lambda_{\mathrm{SLR}}=0.35$. To avoid target-domain model selection, each setting is repeated with three random seeds and evaluated using the final five fixed reporting epochs rather than selecting the best checkpoint on the target domain. Further implementation details are provided in Appendix~\ref{app:preprocessing}

\subsection{Overall Cross-Observation Performance}
\label{sec:overall_performance}


Table~\ref{tab:overall_performance} summarizes the task-level performance under all protocols. MotionQ achieves the highest average accuracy. On the primary gesture protocols W1--W6, which cover multi-factor observation changes and sufficient-coverage settings, MotionQ averages 90.618\% and outperforms CORAL, DANN, WiGRUNT, UniFi, and GesFi by 1.818, 2.368, 10.077, 3.896, and 5.923 percentage points. The gains on challenging observation shifts do not come at the cost of conventional cross-domain recognition, which current methods typically evaluate under a fully observed six-link setting~\cite{liu2024unifi,gu2022wigrunt,zhang2026beyond,xu2025evaluating,fan2025multi,cao2025real}. Such dense link availability, however, is difficult to guarantee in real-world deployments. Under the fully observed W4 setting, MotionQ reaches 99.130\%, 99.478\%, and 97.024\% on cross-direction, -position, and -environment evaluation, respectively. Its comparatively lower performance under cross-environment transfer is consistent with its design scope: methods that seek domain-common representations explicitly suppress static, environment-specific variations resembling style shifts, whereas MotionQ targets changes in the physical observation operator without explicitly enforcing style invariance.

W7 further examines link unavailability after deployment by exhaustively evaluating all 35 possible two- and three-link configurations. MotionQ achieves the best performance under all three settings, reaching 89.943\%, 91.794\%, and 93.742\% under the three-to-two, six-to-two, and six-to-three reductions, respectively. It exceeds the strongest WiFi-specific baseline by $11.981$, $18.398$, and $13.581$ percentage points, respectively. Notably, these multi-link configurations are not explicitly enumerated during single-retained-link training. These results demonstrate that MotionQ generalizes effectively to unseen link subsets and remains strongly robust even under substantial observation loss.

On the PerceptAlign stress test, CORAL obtains the highest three-task average of 78.093\%, while MotionQ reaches 77.397\%, a difference of only 0.696 percentage points. MotionQ still exceeds DANN, WiGRUNT, UniFi, and GesFi by 0.891, 13.311, 1.833, and 2.101 percentage points, respectively. Considering that these tasks involve coordinated full-body activities rather than the compact limb gestures targeted by MotionQ, the results demonstrate useful transfer while also delineating the scope of the two-support motion prior. Across all fourteen task-level metrics, MotionQ achieves the highest overall average of 88.044\%, exceeding the strongest external baseline, DANN, by 2.461 percentage points.

\subsection{Source Coverage and Extrapolative Orientations}
\label{sec:source_coverage}

\begin{figure}[ht]
  \centering
  \includegraphics[width=1\linewidth]{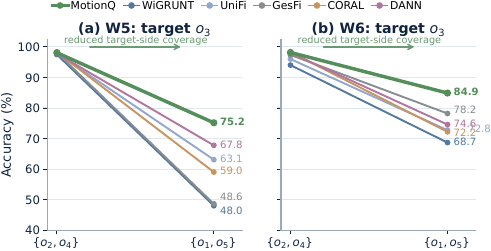}
   \vspace{-0.1in}
  \caption{Impact of target-side operator coverage.}
  \label{fig:target_coverage}
  \vspace{-0.1in}
\end{figure}

We next study how source coverage governs generalization to extrapolative orientations. This question is practically important because the coarse and geometry-sensitive nature of WiFi observations makes it infeasible for collected training data to exhaustively cover the layouts and orientations encountered in real-world deployments. We regard $o_2$--$o_4$ as covered orientations because they lie within the angular range represented by the remaining source orientations, whereas the endpoint orientations $o_1$ and $o_5$ require extrapolations. Across W1--W4, MotionQ exceeds the strongest WiFi-specific baseline by an average of only $1.4$ percentage points on covered orientations, but by $8.5$ percentage points on extrapolative ones. More specifically, across the six endpoint-extrapolation configurations of W1--W3, where environment, user composition, receiver layout, and orientation change simultaneously, MotionQ achieves an average accuracy of $89.2\%$, outperforming WiGRUNT, UniFi, and GesFi by $20.1$, $11.0$, and $16.9$ percentage points, respectively.
The advantage therefore becomes substantially larger when the target operator lies beyond the source coverage, consistent with the common task-observability analysis in Sec.~\ref{sec:task_observability}.

\begin{figure}[ht]
  \centering
  \vspace{-0.12in}
  \includegraphics[width=0.65\linewidth]{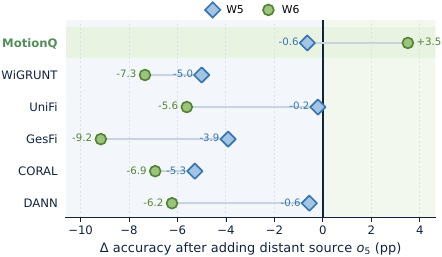}
   \vspace{-0.1in}
  \caption{Impact of adding a geometrically distant source orientation.}
  \label{fig:distant_source}
  \vspace{-0.15in}
\end{figure}



W5 and W6 examine whether additional source orientations improve generalization merely by providing more training data or by better coverage. For target $o_3$, replacing the bracketing pair ${o_2,o_4}$ with the equally sized but distant pair ${o_1,o_5}$ reduces MotionQ by $23.0\%$ and $13.3\%$, whereas the other methods decline by $30.1\%$--$49.6\%$ and $19.2\%$--$26.1\%$ in W5 and W6, respectively (Fig.~\ref{fig:target_coverage}). Fig.~\ref{fig:distant_source} further shows that adding the distant source $o_5$ to $o_2\rightarrow o_1$ decreases all baseline accuracies by $5.6\%$--$9.2\%$ in W6, whereas MotionQ improves by $3.5\%$. When all four remaining source orientations ${o_2,o_3,o_4,o_5}$ are included, MotionQ improves over $o_2\rightarrow o_1$ by $2.2\%$ and $2.8\%$ in W5 and W6, respectively. All baselines remain worse in W6, declining by $4.41\%$--$11.8\%$. In W5, four baselines decline by $0.92\%$--$3.42\%$ (GesFi improving by $5.60\%$). These results show that source diversity helps only when it improves target-side operator coverage. Otherwise, invariant alignment cannot create missing task information and may additionally suppress observable operator-specific cues. MotionQ cannot recover physically absent cues, but it avoids much of this alignment-induced loss by preserving operator-conditioned representations and enforcing equivalence only at the task level.

\begin{figure*}[ht]
  \centering
  \includegraphics[width=0.82\linewidth]{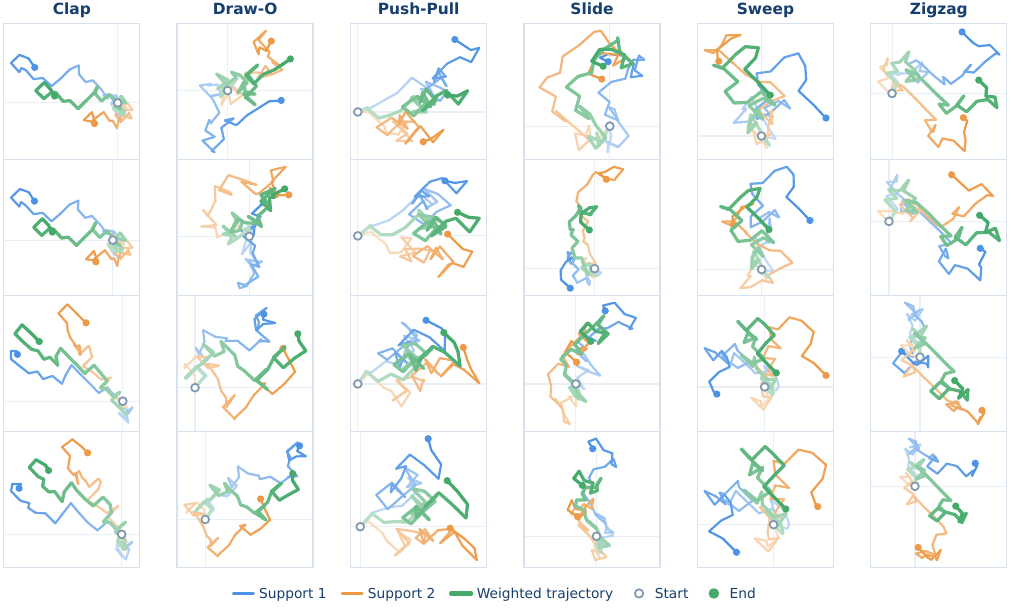}
   \vspace{-0.15in}
  \caption{Emergent gesture-consistent structure in the learned motion
supports.}
  \label{fig:support_trajectories}
  \vspace{-0.1in}
\end{figure*}

\begin{table}[t]
    \centering
    \caption{Ablation studies.}
    \label{tab:ablation}
    \small
    \renewcommand{\arraystretch}{1.05}
    \vspace{-0.15in}
    \textbf{(a) Objective and structural sensitivity}\\[0.5mm]
    \begin{tabular*}{\columnwidth}
        {@{\extracolsep{\fill}}lrrrr@{}}
        \toprule
        $\lambda_{\mathrm{SLR}}$
            & 0 & 0.2 & 0.35 & 0.5 \\
        Acc. (\%)
            & 77.311 & 80.806
            & \textbf{82.100} & 81.050 \\
        \midrule
        $K$
            & 1 & 2 & 3 & -- \\
        Acc. (\%)
            & 80.933 & \textbf{82.100}
            & 81.532 & -- \\
        \midrule
        $H$
            & 2 & 3 & 4 & 6 \\
        Acc. (\%)
            & 81.261 & 81.661
            & \textbf{82.100} & 77.917 \\
        \bottomrule
    \end{tabular*}

    \vspace{1.5mm}

    \textbf{(b) Coordinate intervention}\\[0.5mm]
    \begin{tabular*}{\columnwidth}
        {@{\extracolsep{\fill}}lrrr@{}}
        \toprule
        \textbf{Hyp.}
            & \textbf{Correct (\%)}
            & $\boldsymbol{\Delta}_{\mathrm{swap}}$
            & $\boldsymbol{\Delta}_{\mathrm{collapse}}$ \\
        \midrule
        $h_1$ & 77.300 & +7.056  & +4.483 \\
        $h_2$ & 81.806 & +0.506  & +0.228 \\
        $h_3$ & 83.156 & $-2.556$  & $-1.250$ \\
        $h_4$ & 85.461 & $-26.594$ & $-3.489$ \\
        \midrule
        \textbf{Mixture}
            & 82.611 & $-1.044$ & $-0.472$ \\
        \bottomrule
    \end{tabular*}
    \vspace{-0.2in}
\end{table}

\subsection{Emergent Structure in Motion Supports}
\label{sec:support_interpretability}

Beyond recognition accuracy, we examine whether the learned supports exhibit recurrent gesture-dependent structure. Figure~\ref{fig:support_trajectories} shows four samples per gesture selected with a fixed random seed from W2-$o_3$ and W3-$o_3$ and generated from fixed final-epoch checkpoints. Color progresses from light to dark with time. These diagnostic paths are cumulative sums of latent velocities, not recovered hand trajectories, and the two supports remain exchangeable. Across the two observation settings, several coarse morphologies recur: Push--Pull shows staged directional changes, Draw-O exhibits sustained turning, and Zigzag presents repeated sharp directional turns, while Sweep and Slide exhibit relatively smooth, curved excursions, qualitatively consistent with their arc-like execution around an effective joint center in Widar3.0. Clap remains comparatively directional. Although the exact paths vary across samples, these tendencies recur under different link configurations, consistent with gesture-related latent organization rather than fixed-link appearance. We therefore interpret the figure as qualitative evidence of structured latent motion, not as trajectory reconstruction or validation of specific body-part dynamics. (Appendix~\ref{app:gesture_execution})


\subsection{Ablation Studies}
\label{sec:ablation}

We conduct ablations on W2-$o_1$ using three seeds. The matched ERM backbone obtains an accuracy of $72.233\%$. Incorporating the proposed operator-conditioned support structure raises it to $77.311\%$, while the SLR objective further improves it to $82.100\%$. As shown in Table~\ref{tab:ablation}, the selected weight $\lambda_{\mathrm{SLR}}=0.35$ performs best. Two supports and four operator hypotheses outperform the tested alternatives, indicating that the improvement does not arise merely from increased latent capacity. We further evaluate MotionQ's sensitivity to erroneous coordinates at inference. We consider the paired differences relative to its correct-coordinate result. Swapping the coordinates of Rx2 and Rx5 or collapsing all receiver coordinates to Rx2 reduces the marginalized accuracy by $1.044$ and $0.472$ percentage points, respectively, whereas a uniform $10\,\mathrm{cm}$ coordinate shift changes it by only $+0.106$ percentage points.
Individual hypotheses respond much more strongly: swapping the coordinates improves $h_1$ by $7.056$ but degrades $h_4$ by $26.594$. This contrast arises because the hypotheses represent distinct candidates, so the same coordinate perturbation can move one candidate toward a more informative bistatic projection while moving another away from it. 


\section{Discussion}
\label{sec:limitations}

\textbf{Scope of motion modeling.}
MotionQ is designed for short-duration limb gestures whose discriminative dynamics can be compactly described by a small number of coordinated motion components. Its two-support measure captures barycentric motion, one dominant axis of internal dispersion, and component imbalance, but cannot fully model multiple weakly coupled body parts moving simultaneously along several directions. This boundary is reflected in the PerceptAlign stress test, where MotionQ remains competitive but does not outperform CORAL on average. Therefore, we do not claim that the current two-support model provides a general solution to full-body activity recognition. A natural extension is to develop adaptive or hierarchical motion quotients that introduce additional supports only when required by the observations, enabling applications to complex activities and cross-layout human pose estimation.

\textbf{Physical observability and interpretation.}
MotionQ reduces the task information unnecessarily discarded by universal representation alignment, but it cannot recover motion cues that are physically absent from the target observation operator. Generalization therefore still degrades when the target is poorly covered by the available source operators. Moreover, the learned supports are effective latent velocity components rather than identified body parts or metrically recovered
trajectories. Future work will estimate task observability and prediction uncertainty explicitly, investigate active link or ISAC resource selection to improve target-side coverage, and validate the learned motion structure using independent motion-capture or pose measurements.

\textbf{Deployment assumptions.}
MotionQ assumes known Tx/Rx coordinates and adopts a two-dimensional bistatic geometry. These assumptions are practical in common indoor deployments, where transceivers are typically stationary and their coordinates can be registered once through lightweight calibration~\cite{jia2026breaking}. Extending the current formulation to transceivers that move after calibration and to three-dimensional motion remains future work.



\section{Conclusion}

This work studies cross-observation WiFi gesture recognition, where changes in transceiver geometry alter the physical observation operator rather than merely the appearance of wireless measurements. We establish a common task-observability condition and introduce MotionQ, which learns geometry-conditioned motion quotients while enforcing equivalence only at the task level. Experiments demonstrate strong generalization to unseen layouts and extrapolative orientations, suggesting that preserving operator-conditioned task information is more appropriate than enforcing universal representation invariance.

\bibliographystyle{ACM-Reference-Format}
\bibliography{sample-base}

@String{Computing = "Computing" }

@String{Computer = "{IEEE} Computer" }

@String{Springer = "Springer-Verlag" }

@article{kim2026simplified,
  title={A simplified wearable device powered by a generative EMG network for hand-gesture recognition and gait prediction},
  author={Kim, Kyun Kyu and Zaluska, Tomasz J and Skov, Six and Lee, Yeongjun and Park, Hyunchang and Zhong, Donglai and Khatib, Muhammad and Nishio, Yuya and Jiang, Yuanwen and Delp, Scott L and others},
  journal={Nature Sensors},
  volume={1},
  number={1},
  pages={27--38},
  year={2026},
  publisher={Nature Publishing Group UK London}
}

@ARTICLE{9363693,
  author={IEEE},
  journal={IEEE Std 802.11-2025}, 
  title={IEEE Standard for Information Technology -- Telecommunications and Information Exchange Between Systems Local and Metropolitan Area Networks -- Specific Requirements - Part 11: Wireless LAN Medium Access Control (MAC) and Physical Layer (PHY) Specifications - Amendment 4: Enhancements for Wireless LAN Sensing}, 
  year={2025},
  volume={},
  number={},
  pages={1-4379},
  doi={10.1109/IEEESTD.2025.11184214}}

@article{li2025cross,
  title={Cross-domain gesture recognition via WiFi signals with deep learning},
  author={Li, Baogang and Chen, Jiale and Yu, Xinlong and Yang, Zhi and Zhang, Jingxi},
  journal={Ad Hoc Networks},
  volume={166},
  pages={103654},
  year={2025},
  publisher={Elsevier}
}

@article{yan2025wi,
  title={Wi-sfdagr: Wifi-based cross-domain gesture recognition via source-free domain adaptation},
  author={Yan, Huan and Zhang, Xiang and Huang, Jinyang and Feng, Yuanhao and Li, Meng and Wang, Anzhi and Ou, Weihua and Wang, Hongbing and Liu, Zhi},
  journal={IEEE Internet of Things Journal},
  volume={12},
  number={13},
  pages={24159--24173},
  year={2025},
  publisher={IEEE}
}

@article{li2025wilife,
  title={Wilife: Long-term daily status monitoring and habit mining of the elderly leveraging ubiquitous wi-fi signals},
  author={Li, Shengjie and Liu, Zhaopeng and Lv, Qin and Zou, Yanyan and Zhang, Yue and Zhang, Daqing},
  journal={ACM Transactions on Computing for Healthcare},
  volume={6},
  number={1},
  pages={1--29},
  year={2025},
  publisher={ACM New York, NY}
}

@article{zhang2026beyond,
  title={Beyond physical labels: Redefining domains for robust WiFi-based gesture recognition},
  author={Zhang, Xiang and Yan, Huan and Huang, Jinyang and Liu, Bin and Feng, Yuanhao and Liu, Jianchun and Li, Meng and Zhang, Fusang and Liu, Zhi},
  journal={Proceedings of the ACM on Interactive, Mobile, Wearable and Ubiquitous Technologies},
  volume={10},
  number={1},
  pages={1--27},
  year={2026},
  publisher={ACM New York, NY, USA}
}

@article{miao2025wi,
  title={Wi-Fi sensing techniques for human activity recognition: Brief survey, potential challenges, and research directions},
  author={Miao, Fucheng and Huang, Youxiang and Lu, Zhiyi and Ohtsuki, Tomoaki and Gui, Guan and Sari, Hikmet},
  journal={ACM Computing Surveys},
  volume={57},
  number={5},
  pages={1--30},
  year={2025},
  publisher={ACM New York, NY}
}

@article{yin2022fewsense,
  title={FewSense, towards a scalable and cross-domain Wi-Fi sensing system using few-shot learning},
  author={Yin, Guolin and Zhang, Junqing and Shen, Guanxiong and Chen, Yingying},
  journal={IEEE Transactions on Mobile Computing},
  volume={23},
  number={1},
  pages={453--468},
  year={2022},
  publisher={IEEE}
}

@article{sheng2024metaformer,
  title={MetaFormer: Domain-adaptive WiFi sensing with only one labelled target sample},
  author={Sheng, Biyun and Han, Rui and Xiao, Fu and Guo, Zhengxin and Gui, Linqing},
  journal={Proceedings of the ACM on Interactive, Mobile, Wearable and Ubiquitous Technologies},
  volume={8},
  number={1},
  pages={1--27},
  year={2024},
  publisher={ACM New York, NY, USA}
}

@article{feng2022wi,
  title={Wi-learner: Towards one-shot learning for cross-domain wi-fi based gesture recognition},
  author={Feng, Chao and Wang, Nan and Jiang, Yicheng and Zheng, Xia and Li, Kang and Wang, Zheng and Chen, Xiaojiang},
  journal={Proceedings of the ACM on Interactive, Mobile, Wearable and Ubiquitous Technologies},
  volume={6},
  number={3},
  pages={1--27},
  year={2022},
  publisher={ACM New York, NY, USA}
}

@article{zhang2021widar3,
  title={Widar3. 0: Zero-effort cross-domain gesture recognition with Wi-Fi},
  author={Zhang, Yi and Zheng, Yue and Qian, Kun and Zhang, Guidong and Liu, Yunhao and Wu, Chenshu and Yang, Zheng},
  journal={IEEE Transactions on Pattern Analysis and Machine Intelligence},
  volume={44},
  number={11},
  pages={8671--8688},
  year={2021},
  publisher={IEEE}
}

@article{gao2021towards,
  title={Towards position-independent sensing for gesture recognition with Wi-Fi},
  author={Gao, Ruiyang and Zhang, Mi and Zhang, Jie and Li, Yang and Yi, Enze and Wu, Dan and Wang, Leye and Zhang, Daqing},
  journal={Proceedings of the ACM on Interactive, Mobile, Wearable and Ubiquitous Technologies},
  volume={5},
  number={2},
  pages={1--28},
  year={2021},
  publisher={ACM New York, NY, USA}
}

@article{gu2022wigrunt,
  title={WiGRUNT: WiFi-enabled gesture recognition using dual-attention network},
  author={Gu, Yu and Zhang, Xiang and Wang, Yantong and Wang, Meng and Yan, Huan and Ji, Yusheng and Liu, Zhi and Li, Jianhua and Dong, Mianxiong},
  journal={IEEE transactions on human-machine systems},
  volume={52},
  number={4},
  pages={736--746},
  year={2022},
  publisher={IEEE}
}

@article{liu2023wisr,
  title={WiSR: Wireless domain generalization based on style randomization},
  author={Liu, Shijia and Chen, Zhenghua and Wu, Min and Liu, Chang and Chen, Liangyin},
  journal={IEEE Transactions on Mobile Computing},
  volume={23},
  number={5},
  pages={4520--4532},
  year={2023},
  publisher={IEEE}
}

@article{liu2024unifi,
  title={Unifi: A unified framework for generalizable gesture recognition with wi-fi signals using consistency-guided multi-view networks},
  author={Liu, Yan and Yu, Anlan and Wang, Leye and Guo, Bin and Li, Yang and Yi, Enze and Zhang, Daqing},
  journal={Proceedings of the ACM on Interactive, Mobile, Wearable and Ubiquitous Technologies},
  volume={7},
  number={4},
  pages={1--29},
  year={2024},
  publisher={ACM New York, NY, USA}
}

@article{jia2026breaking,
  title={Breaking Coordinate Overfitting: Geometry-Aware WiFi Sensing for Cross-Layout 3D Pose Estimation},
  author={Jia, Songming and Lu, Yan and Liu, Bin and Zhang, Xiang and Zhao, Peng and Tang, Xinmeng and Wei, Yelin and Huang, Jinyang and Yan, Huan and Liu, Zhi},
  journal={arXiv preprint arXiv:2601.12252},
  year={2026}
}

@article{zhang2025wiopen,
  title={Wiopen: A robust wi-fi-based open-set gesture recognition framework},
  author={Zhang, Xiang and Huang, Jinyang and Yan, Huan and Feng, Yuanhao and Zhao, Peng and Zhuang, Guohang and Liu, Zhi and Liu, Bin},
  journal={IEEE Transactions on Human-Machine Systems},
  volume={55},
  number={2},
  pages={234--245},
  year={2025},
  publisher={IEEE}
}

@article{he2026beamforming,
  title={Beamforming-enabled integrated sensing and communication over commodity multi-user Wi-Fi},
  author={He, Yinghui and Xu, Mingming and Chen, Zhe and Xiao, Fu and Luo, Jun},
  journal={IEEE Transactions on Mobile Computing},
  year={2026},
  publisher={IEEE}
}

@article{han2026rayloc,
  title={RayLoc: Wireless indoor localization via fully differentiable ray-tracing},
  author={Han, Xueqiang and Zheng, Tianyue and Hu, Menglan and Cai, Chao and Han, Tony Xiao and Luo, Jun},
  journal={Proceedings of the ACM on Interactive, Mobile, Wearable and Ubiquitous Technologies},
  volume={10},
  number={2},
  pages={1--24},
  year={2026},
  publisher={ACM New York, NY, USA}
}

@article{chen2026wipihgr,
  title={WiPIHGR: A Physically-Guided Feature Learning Approach for Position-Independent Wi-Fi Gesture Recognition},
  author={Chen, Xingcan and Li, Chenglin and Meng, Wei and Xiao, Wendong},
  journal={IEEE Transactions on Mobile Computing},
  year={2026},
  publisher={IEEE}
}

@article{he2025beam,
  title={Beam-Fi: Integrated sensing and communication via MU-MIMO upon commodity Wi-Fi},
  author={He, Yinghui and Xu, Mingming and Chen, Zhe and Xiao, Fu and Luo, Jun},
  journal={Proceedings of the ACM on Interactive, Mobile, Wearable and Ubiquitous Technologies},
  volume={9},
  number={3},
  pages={1--22},
  year={2025},
  publisher={ACM New York, NY, USA}
}

@article{li2020wihf,
  title={WiHF: Gesture and user recognition with WiFi},
  author={Li, Chenning and Liu, Manni and Cao, Zhichao},
  journal={IEEE Transactions on Mobile Computing},
  volume={21},
  number={2},
  pages={757--768},
  year={2020},
  publisher={IEEE}
}

@inproceedings{fan2026sense,
  title={Sense with polyface mirror: Enhancing Wi-Fi sensing diversity via programmable metasurfaces},
  author={Fan, Long and He, Yinghui and Xie, Lei and Zhang, Serene and Luo, Jun},
  booktitle={Proceedings of the 2026 ACM/IEEE International Conference on Embedded Artificial Intelligence and Sensing Systems},
  pages={214--228},
  year={2026}
}

@inproceedings{li2024uwb,
  title={Uwb-fi: Pushing wi-fi towards ultra-wideband for fine-granularity sensing},
  author={Li, Xin and Wang, Hongbo and Chen, Zhe and Jiang, Zhiping and Luo, Jun},
  booktitle={Proceedings of the 22nd Annual International Conference on Mobile Systems, Applications and Services},
  pages={42--55},
  year={2024}
}

@inproceedings{hu2023muse,
  title={MUSE-Fi: Contactless muti-person sensing exploiting near-field Wi-Fi channel variation},
  author={Hu, Jingzhi and Zheng, Tianyue and Chen, Zhe and Wang, Hongbo and Luo, Jun},
  booktitle={Proceedings of the 29th annual international conference on mobile computing and networking},
  pages={1--15},
  year={2023}
}

@inproceedings{li2025muceiver,
  title={$\mu$Ceiver-Fi: Exploiting Spectrum Resources of Multi-Link Receiver for Fine-Granularity Wi-Fi Sensing},
  author={Li, Xin and He, Yinghui and Luo, Jun},
  booktitle={Proceedings of the 31st Annual International Conference on Mobile Computing and Networking},
  pages={1045--1059},
  year={2025}
}

@inproceedings{hu2025poison,
  title={Poison to Cure: Privacy-preserving Wi-Fi Multi-User Sensing via Data Poisoning},
  author={Hu, Jingzhi and Li, Xin and Gan, Jin and Luo, Jun},
  booktitle={Proceedings of the 31st Annual International Conference on Mobile Computing and Networking},
  pages={47--62},
  year={2025}
}

@article{wang2025vr,
  title={Vr-fi: Positioning and recognizing hand gestures via vr-embedded wi-fi sensing},
  author={Wang, Hongbo and Li, Xin and Li, Jiachun and Zhu, Haojin and Luo, Jun},
  journal={IEEE Transactions on Mobile Computing},
  volume={24},
  number={9},
  pages={8287--8300},
  year={2025},
  publisher={IEEE}
}

@inproceedings{abdulatif2017real,
  title={Real-time capable micro-Doppler signature decomposition of walking human limbs},
  author={Abdulatif, Sherif and Aziz, Fady and Kleiner, Bernhard and Schneider, Urs},
  booktitle={2017 IEEE Radar Conference},
  pages={1093--1098},
  year={2017},
  organization={IEEE}
}

@article{rohrer2004submovements,
  title={Submovements grow larger, fewer, and more blended during stroke recovery},
  author={Rohrer, Brandon and Fasoli, Susan and Krebs, Hermano Igo and Volpe, Bruce and Frontera, Walter R and Stein, Joel and Hogan, Neville},
  journal={Motor control},
  volume={8},
  number={4},
  pages={472--483},
  year={2004},
  publisher={Human Kinetics, Inc.}
}

@misc{tang2019kinematic,
  title={Kinematic synergy of multi-DOF movement in upper limb and its application for rehabilitation exoskeleton motion planning. Front Neurorobot},
  author={Tang, S and Chen, L and Barsotti, M and Hu, L and Li, Y and Wu, X and Bai, L and Frisoli, A and Hou, W},
  year={2019},
  publisher={Neurorobot}
}

@article{li2025cross1,
  title={Cross-Domain Multi-Person Human Activity Recognition via Near-Field Wi-Fi Sensing},
  author={Li, Xin and Hu, Jingzhi and He, Yinghui and Wang, Hongbo and Gan, Jin and Luo, Jun},
  journal={arXiv preprint arXiv:2510.17816},
  year={2025}
}

@inproceedings{sun2016deep,
  title={Deep coral: Correlation alignment for deep domain adaptation},
  author={Sun, Baochen and Saenko, Kate},
  booktitle={European conference on computer vision},
  pages={443--450},
  year={2016},
  organization={Springer}
}

@article{ganin2016domain,
  title={Domain-adversarial training of neural networks},
  author={Ganin, Yaroslav and Ustinova, Evgeniya and Ajakan, Hana and Germain, Pascal and Larochelle, Hugo and Laviolette, Fran{\c{c}}ois and March, Mario and Lempitsky, Victor},
  journal={Journal of machine learning research},
  volume={17},
  number={59},
  pages={1--35},
  year={2016}
}

@article{fan2025multi,
  title={Multi-source domain generalization for csi-based human activity recognition},
  author={Fan, Tianqi and Qiu, Sen and Gong, Wei and Fang, Yuguang},
  journal={IEEE Transactions on Mobile Computing},
  volume={24},
  number={10},
  pages={11034--11045},
  year={2025},
  publisher={IEEE}
}

@article{liu2025efficient,
  title={Efficient one-shot gesture recognition for WiFi ISAC via aug-meta learning},
  author={Liu, Jianwei and Yuan, Jiantao and Yu, Guanding and Han, Jinsong},
  journal={IEEE Journal on Selected Areas in Communications},
  volume={43},
  number={11},
  pages={3766--3781},
  year={2025},
  publisher={IEEE}
}

@article{cao2025real,
  title={Real-time cross-domain gesture and user identification via COTS WiFi},
  author={Cao, Chenhong and Ding, Yue and Dai, Miaoling and Gong, Wei and Zhao, Xibin},
  journal={IEEE Transactions on Mobile Computing},
  volume={24},
  number={6},
  pages={5124--5137},
  year={2025},
  publisher={IEEE}
}

@article{xu2025evaluating,
  title={Evaluating self-supervised learning for WiFi CSI-based human activity recognition},
  author={Xu, Ke and Wang, Jiangtao and Zhu, Hongyuan and Zheng, Dingchang},
  journal={ACM Transactions on Sensor Networks},
  volume={21},
  number={2},
  pages={1--38},
  year={2025},
  publisher={ACM New York, NY}
}

@article{gao2022towards,
  title={Towards robust gesture recognition by characterizing the sensing quality of WiFi signals},
  author={Gao, Ruiyang and Li, Wenwei and Xie, Yaxiong and Yi, Enze and Wang, Leye and Wu, Dan and Zhang, Daqing},
  journal={Proceedings of the ACM on Interactive, Mobile, Wearable and Ubiquitous Technologies},
  volume={6},
  number={1},
  pages={1--26},
  year={2022},
  publisher={ACM New York, NY, USA}
}

@article{chen2024wignn,
  title={WiGNN: WiFi-based cross-domain gesture recognition inspired by dynamic topology structure},
  author={Chen, Yinan and Huang, Xiaoxia},
  journal={IEEE Wireless Communications},
  volume={31},
  number={3},
  pages={249--256},
  year={2024},
  publisher={IEEE}
}

@article{wang2022airfi,
  title={AirFi: Empowering WiFi-based passive human gesture recognition to unseen environment via domain generalization},
  author={Wang, Dazhuo and Yang, Jianfei and Cui, Wei and Xie, Lihua and Sun, Sumei},
  journal={IEEE Transactions on Mobile Computing},
  volume={23},
  number={2},
  pages={1156--1168},
  year={2022},
  publisher={IEEE}
}

@inproceedings{zhang2026wi,
  title={Wi-CBR: Salient-aware adaptive WiFi sensing for cross-domain behavior recognition},
  author={Zhang, Ruobei and Tang, Shengeng and Yan, Huan and Zhang, Xiang and Guo, Jiabao},
  booktitle={Proceedings of the AAAI Conference on Artificial Intelligence},
  volume={40},
  number={2},
  pages={1552--1560},
  year={2026}
}

@inproceedings{johansson2019support,
  title={Support and invertibility in domain-invariant representations},
  author={Johansson, Fredrik D and Sontag, David and Ranganath, Rajesh},
  booktitle={The 22nd International Conference on Artificial Intelligence and Statistics},
  pages={527--536},
  year={2019},
  organization={PMLR}
}

@inproceedings{zhao2019learning,
  title={On learning invariant representations for domain adaptation},
  author={Zhao, Han and Des Combes, Remi Tachet and Zhang, Kun and Gordon, Geoffrey},
  booktitle={International conference on machine learning},
  pages={7523--7532},
  year={2019},
  organization={PMLR}
}

@article{liang2023factorized,
  title={Factorized contrastive learning: Going beyond multi-view redundancy},
  author={Liang, Paul Pu and Deng, Zihao and Ma, Martin Q and Zou, James Y and Morency, Louis-Philippe and Salakhutdinov, Ruslan},
  journal={Advances in Neural Information Processing Systems},
  volume={36},
  pages={32971--32998},
  year={2023}
}

@article{chen2023cross,
  title={Cross-domain WiFi sensing with channel state information: A survey},
  author={Chen, Chen and Zhou, Gang and Lin, Youfang},
  journal={ACM Computing Surveys},
  volume={55},
  number={11},
  pages={1--37},
  year={2023},
  publisher={ACM New York, NY}
}

@article{li2021tilted,
  title={Tilted empirical risk minimization},
  author={Li, Tian and Beirami, Ahmad and Sanjabi, Maziar and Smith, Virginia},
  journal={ICLR 2021},
  year={2021}
}

@article{monteiller2019alleviating,
  title={Alleviating label switching with optimal transport},
  author={Monteiller, Pierre and Claici, Sebastian and Chien, Edward and Mirzazadeh, Farzaneh and Solomon, Justin M and Yurochkin, Mikhail},
  journal={Advances in Neural Information Processing Systems},
  volume={32},
  year={2019}
}

@inproceedings{he2016deep,
  title={Deep residual learning for image recognition},
  author={He, Kaiming and Zhang, Xiangyu and Ren, Shaoqing and Sun, Jian},
  booktitle={Proceedings of the IEEE conference on computer vision and pattern recognition},
  pages={770--778},
  year={2016}
}

@article{zhang2023wital,
  title={Wital: A COTS WiFi devices based vital signs monitoring system using NLOS sensing model},
  author={Zhang, Xiang and Gu, Yu and Yan, Huan and Wang, Yantong and Dong, Mianxiong and Ota, Kaoru and Ren, Fuji and Ji, Yusheng},
  journal={IEEE Transactions on Human-Machine Systems},
  volume={53},
  number={3},
  pages={629--641},
  year={2023},
  publisher={IEEE}
}

@inproceedings{chang2026wirainbow,
  title={WiRainbow: Single-Antenna Direction-Aware Wi-Fi Sensing via Dispersion Effect},
  author={Chang, Zhaoxin and Xiao, Shuguang and Zhang, Fusang and Ma, Xujun and Jouaber, Badii and Zhang, Qingfeng and Zhang, Daqing},
  booktitle={Proceedings of the 2026 ACM/IEEE International Conference on Embedded Artificial Intelligence and Sensing Systems},
  pages={862--875},
  year={2026}
}

@incollection{zhang2023wifi,
  title={Wifi/4g/5g based wireless sensing: Theories, applications and future directions},
  author={Zhang, Daqing and Niu, Kai and Xiong, Jie and Zhang, Fusang and Wang, Xuanzhi},
  booktitle={Integrated Sensing and Communications},
  pages={387--417},
  year={2023},
  publisher={Springer}
}

@inproceedings{deng2009imagenet,
  title={Imagenet: A large-scale hierarchical image database},
  author={Deng, Jia and Dong, Wei and Socher, Richard and Li, Li-Jia and Li, Kai and Fei-Fei, Li},
  booktitle={2009 IEEE conference on computer vision and pattern recognition},
  pages={248--255},
  year={2009},
  organization={Ieee}
}

@article{zeng2019farsense,
  title={FarSense: Pushing the range limit of WiFi-based respiration sensing with CSI ratio of two antennas},
  author={Zeng, Youwei and Wu, Dan and Xiong, Jie and Yi, Enze and Gao, Ruiyang and Zhang, Daqing},
  journal={Proceedings of the ACM on Interactive, Mobile, Wearable and Ubiquitous Technologies},
  volume={3},
  number={3},
  pages={1--26},
  year={2019},
  publisher={ACM New York, NY, USA}
}

@article{yan2026diffloc+,
  title={DiffLoc+: Towards Robust WiFi Hidden Camera Localization Based on Electromagnetic Diffraction},
  author={Yan, Huan and Liu, Jian and Zhang, Xiang and Liu, Zhi and Liu, Bin and Li, Meng and Gong, Zheng and Gao, Ming and Zhang, Fusang},
  journal={IEEE Journal on Selected Areas in Communications},
  year={2026},
  publisher={IEEE}
}

@inproceedings{wei2025source,
  title={Source-Free Domain Adaptation via Perceptual Semantic Decoupling for WiFi Gesture Recognition},
  author={Wei, Yelin and Zhang, Xiang and Liu, Bin and Jia, Songming and Huang, Jinyang and Liu, Zhi and Yan, Huan},
  booktitle={GLOBECOM 2025-2025 IEEE Global Communications Conference},
  pages={4089--4094},
  year={2025},
  organization={IEEE}
}

@article{tan2025wimap,
  title={Wimap: Autonomous wi-fi mapping for device-free tracking in smart homes},
  author={Tan, Renrui and Hong, Tu and Tian, Yichen and Tong, Xinyu and Chen, Sheng and Liu, Xiulong and Xie, Xin and Qu, Wenyu},
  journal={Proceedings of the ACM on Interactive, Mobile, Wearable and Ubiquitous Technologies},
  volume={9},
  number={4},
  year={2025},
  publisher={ACM New York, NY, USA}
}

@article{meng2025metatrack,
  title={Metatrack: Enabling wi-fi device free tracking in complex scenarios},
  author={Meng, Xuanqi and Ge, Weiping and Tian, Yichen and Tong, Xinyu and Liu, Xiulong and Xie, Xin and Qu, Wenyu},
  journal={Proceedings of the ACM on Interactive, Mobile, Wearable and Ubiquitous Technologies},
  volume={9},
  number={4},
  pages={1--26},
  year={2025},
  publisher={ACM New York, NY, USA}
}

@article{zhao2025baton,
  title={Baton: Compensate for Missing Wi-Fi Features for Practical Device-Free Tracking},
  author={Zhao, Yiming and Meng, Xuanqi and Tong, Xinyu and Liu, Xiulong and Xie, Xin and Qu, Wenyu},
  journal={IEEE Transactions on Mobile Computing},
  volume={24},
  number={9},
  pages={9238--9254},
  year={2025},
  publisher={IEEE}
}

@inproceedings{zhang2025camlopa,
  title={Camlopa: A hidden wireless camera localization framework via signal propagation path analysis},
  author={Zhang, Xiang and Zhang, Jie and Ma, Zehua and Huang, Jinyang and Li, Meng and Yan, Huan and Zhao, Peng and Zhang, Zijian and Liu, Bin and Guo, Qing and others},
  booktitle={2025 IEEE symposium on security and privacy (SP)},
  pages={3653--3671},
  year={2025},
  organization={IEEE}
}

@article{feng2025imbalanced,
  title={Imbalanced semi-supervised learning for wifi gesture recognition via dynamic threshold-based spatio-temporal attention networks},
  author={Feng, Qihua and Duan, Chunhui and Xue, Jiawei and Li, Chaozhuo and Huang, Feiran and Zhang, Xi and Weng, Jian and Yu, Philip S},
  journal={IEEE Transactions on Mobile Computing},
  volume={25},
  number={1},
  pages={483--499},
  year={2025},
  publisher={IEEE}
}

\appendix

\begin{table*}[ht]
    \centering
    \caption{\textbf{Detailed Widar3.0 evaluation protocols.}
    ``LOO-Ori.'', ``LOO-Pos.'' and ``LOO-Env.'' denote leave-one-orientation,
    leave-one-position and leave-one-environment evaluation, respectively.}
    \label{tab:widar_protocols}
    \scriptsize
    \setlength{\tabcolsep}{4.1pt}
    \renewcommand{\arraystretch}{1.10}
    \begin{tabular}{@{}p{0.055\textwidth}
                        p{0.205\textwidth}
                        p{0.205\textwidth}
                        p{0.205\textwidth}
                        p{0.255\textwidth}@{}}
        \toprule
        \textbf{ID} &
        \textbf{Source observations} &
        \textbf{Target observations} &
        \textbf{Orientation / position protocol} &
        \textbf{Purpose} \\
        \midrule

        W1 &
        Scene~1, users 1--9;
        layouts $[2,5]$ and $[2,6]$ &
        Scene~2 and Scene~3, eight unseen users
        (four per scene);
        unseen layouts $[1,3]$ and $[1,4]$ &
        Five-fold LOO-Ori.; the remaining four orientations train each fold &
        Simultaneous cross-environment, cross-user, cross-layout, and
        cross-orientation generalization with two-link observations \\
        
        W2 &
        Scene~1, users 1--9;
        layouts $[1,2]$ and $[2,3]$ &
        Scene~2 and Scene~3, eight unseen users
        (four per scene);
        unseen layout $[2,5]$ &
        Five-fold LOO-Ori. &
        Multi-factor cross-observation generalization to an unseen
        two-link observation geometry \\
        
        W3 &
        Scene~1, users 1--9;
        three-link layout $[2,5,6]$ &
        Scene~2 and Scene~3, eight unseen users
        (four per scene);
        complementary three-link layout $[1,3,4]$ &
        Five-fold LOO-Ori. &
        Multi-factor generalization between complementary three-link
        observation geometries \\

        W4 &
        All six links $[1,2,3,4,5,6]$ &
        All six links $[1,2,3,4,5,6]$ &
        W4-D: five-fold LOO-Ori.;
        W4-P: five-fold LOO-Pos;
        W4-E: five-fold LOO-Env. &
        Conventional cross-domain recognition under sufficiently covered
        observations \\

        W5 &
        Fixed layout $[1,5]$ &
        Fixed layout $[1,5]$ &
        Target $o_1$:
        $\{o_2\}$, $\{o_3\}$, $\{o_4\}$, $\{o_5\}$,
        $\{o_2,o_5\}$, or $\{o_2,o_3,o_4,o_5\}$ as sources;
        target $o_3$:
        $\{o_1\}$, $\{o_2\}$, $\{o_4\}$, $\{o_5\}$,
        $\{o_2,o_4\}$, $\{o_1,o_5\}$, or
        $\{o_1,o_2,o_4,o_5\}$ as sources &
        Isolate source-operator coverage and orientation extrapolation
        without introducing layout changes \\

        W6 &
        Source environments/user composition;
        layouts $[1,2]$ and $[2,3]$ &
        Different environments/user composition;
        unseen layout $[2,5]$ &
        Same 13 source--target orientation configurations as W5 &
        Test source coverage and heterogeneous-source effects while
        simultaneously changing environment, user composition, and layout \\

        W7 &
        (i) $[2,5,6]$, or
        (ii) all six links;
        source orientations $o_2$--$o_5$ &
        Target environments/user composition and target orientation $o_1$;
        all real two- or three-link subsets &
        One frozen model per seed; no retraining or target adaptation &
        Link unavailability:
        $3\!\rightarrow\!2$ evaluates all $\binom{6}{2}=15$ subsets;
        $6\!\rightarrow\!2$ evaluates all 15 subsets; and
        $6\!\rightarrow\!3$ evaluates all $\binom{6}{3}=20$ subsets \\
        \bottomrule
    \end{tabular}
\end{table*}

\begin{table*}[ht]
    \centering
    \caption{\textbf{Detailed PerceptAlign stress-test protocols.}}
    \scriptsize
    \setlength{\tabcolsep}{5pt}
    \renewcommand{\arraystretch}{1.10}
    \begin{tabular}{@{}p{0.06\textwidth}
                        p{0.2\textwidth}
                        p{0.25\textwidth}
                        p{0.21\textwidth}
                        p{0.20\textwidth}@{}}
        \toprule
        \textbf{ID} &
        \textbf{Source observations} &
        \textbf{Target observations} &
        \textbf{Orientation protocol} &
        \textbf{Purpose} \\
        \midrule

        P1 &
        Scene~2, users 1--9, logical layout $[1,2,3]$ &
        Scene~3A/B/C, users 1--6, logical layout $[1,2,3]$;
        each target configuration uses its own physical Rx coordinates &
        Three-fold LOO-Ori.:
        $\{o_2,o_3\}\!\rightarrow o_1$,
        $\{o_1,o_3\}\!\rightarrow o_2$, and
        $\{o_1,o_2\}\!\rightarrow o_3$ &
        Cross-scene and cross-geometry full-body motion stress test \\
        
        P2 &
        Scene~2, users 1--9;
        source layouts $[1,3]$ and $[2,3]$ &
        Scene~3A/B/C, users 1--6;
        target layouts $[1,3]$, $[2,3]$, and unseen $[1,2]$ &
        Same three-fold LOO-Ori. as P1 &
        Simultaneous scene, receiver-layout, and orientation changes \\
        
        P3 &
        Scene~2, users 1--9;
        source layouts $[1,3]$ and $[2,3]$ &
        Scene~3A/B/C, users 1--6;
        target layouts $[1,3]$, $[2,3]$, and unseen $[1,2]$ &
        Extreme transfers:
        $o_3\!\rightarrow o_1$ and $o_1\!\rightarrow o_3$ &
        Cross-scene and cross-layout stress test under large
        orientation extrapolation \\
        \bottomrule
    \end{tabular}
\end{table*}

\section{Detailed Evaluation Protocols}
\label{app:protocols}

This appendix provides the complete source--target configurations used in Sec.~\ref{sec:experimental_setup}. Unless otherwise specified, all target observations are strictly excluded from training, model selection, and hyperparameter tuning. Each configuration is repeated with three fixed random seeds (1111, 2026, and 3456). For protocols containing multiple target layouts or scenes, each target is evaluated independently and the reported task-level accuracy is their equal-weight macro-average. The W4-D, W4-P, and W4-E results of WiGRUNT, UniFi, and GesFi are taken from their corresponding papers.

\textbf{Widar3.0 protocols.}
Widar3.0 contains five user orientations, $o_1$--$o_5$, corresponding to $135^\circ$, $90^\circ$, $45^\circ$, $0^\circ$, and $-45^\circ$, respectively. For the multi-factor protocols W1--W3, all source data are collected in Scene~1 from nine users, whereas testing is performed in Scene~2 and Scene~3 using eight different users, with four users in each target scene. The source and target user sets are disjoint. Thus, W1--W3 simultaneously change the collection environment, user population, receiver layout, and target orientation. The transceiver coordinates are shared across the three Widar3.0 scenes; cross-layout therefore refers to changing the subset of receivers participating in the observation rather than relocating the physical receiver coordinates. We denote a receiver layout by $[i,j,\ldots]$, where the indices specify the receivers available to the model.

\textbf{PerceptAlign protocols.}
PerceptAlign is evaluated independently from Widar3.0. We use Scene~2 as the source collection and Scene~3 as the target collection. Scene~2 contains users 1--9, whereas Scene~3 contains users 1--6; the two collections therefore share user identities, and we do not characterize these protocols as strict cross-user generalization. Scene~3 contains three receiver configurations, denoted Scene~3A, Scene~3B, and Scene~3C, each with its own physical Rx coordinates. We remap the three physical receivers in each configuration to logical links $[1,2,3]$ for protocol specification, while MotionQ always uses the corresponding physical coordinates when constructing the bistatic observation operator. 
We retain four base activities---Stretch, Lunge, Squat, and Jump, and do not treat direction-specific execution variants as separate classes. We define recognition at the base-activity level because direction-specific executions of the same activity are related primarily by a global rotation and are therefore task-equivalent under our cross-observation formulation. This label definition is applied consistently to all methods.


\section{Model Details}
\label{app:modetail}

\textbf{Active-layout coordinate encoding.}
For an active link set $\mathcal V$, we compute its receiver centroid, layout center, and layout scale as
\begin{equation}
    \bar R_{\mathcal V}
    =
    \frac{1}{|\mathcal V|}
    \sum_{i\in\mathcal V}R_i,
    \qquad
    c_{\mathcal V}
    =
    \frac{T+\bar R_{\mathcal V}}{2},
\end{equation}
\begin{equation}
    L_{\mathcal V}
    =
    \sqrt{
        \frac{1}{|\mathcal V|}
        \sum_{i\in\mathcal V}
        \lVert R_i-T\rVert_2^2
    }.
\end{equation}
The first frame axis points from the transmitter to the active receiver centroid:
\begin{equation}
    e_{\mathcal V}
    =
    \frac{\bar R_{\mathcal V}-T}
    {\lVert\bar R_{\mathcal V}-T\rVert_2},
    \qquad
    U_{\mathcal V}
    =
    \begin{bmatrix}
        e_{\mathcal V,x} & -e_{\mathcal V,y}\\
        e_{\mathcal V,y} &  e_{\mathcal V,x}
    \end{bmatrix}.
\end{equation}
The layout-relative receiver encoding is
\begin{equation}
    \psi_{\mathcal V}(R_i)
    =
    \frac{
        U_{\mathcal V}^{\top}(R_i-c_{\mathcal V})
    }{
        s_{\mathrm{coord}}L_{\mathcal V}
    },
\end{equation}
where $s_{\mathrm{coord}}$ is a fixed numerical scale.

Let $\alpha_h$ denote the orientation anchor of hypothesis $h$. The operator network predicts a shared angular residual $\Delta_{\mathrm{glob}}$, an anchor-specific residual $\Delta_h$, and a local position $p_{\mathrm{loc}}$. The resulting parameters are
\begin{equation}
    \theta_h
    =
    \alpha_h+\Delta_{\mathrm{glob}}+\Delta_h,
    \qquad
    o_h
    =
    U_{\mathcal V}
    \begin{bmatrix}
        \cos\theta_h\\
        \sin\theta_h
    \end{bmatrix},
\end{equation}
\begin{equation}
    p
    =
    c_{\mathcal V}
    +
    L_{\mathcal V}U_{\mathcal V}p_{\mathrm{loc}}.
\end{equation}
The residuals and $p_{\mathrm{loc}}$ are bounded using $\tanh$. Only the output layers of the corresponding prediction heads are initialized to zero, so the initial orientations coincide with their anchors and the initial position is $c_{\mathcal V}$.

\textbf{Multiplicative bistatic conditioning.}
For the dimensionless bistatic coordinate
$\tilde q_{hi}=\lambda q_{hi}/2$, the directional basis is
\begin{equation}
    \phi(\tilde q_{hi})
    =
    \begin{bmatrix}
        \tilde q_{hi,x}\\
        \tilde q_{hi,y}\\
        \tilde q_{hi,x}^2-\tilde q_{hi,y}^2\\
        2\tilde q_{hi,x}\tilde q_{hi,y}
    \end{bmatrix}.
\end{equation}
Each $E_k$ consists of two bias-free linear layers of dimensions
$512\rightarrow256\rightarrow192$, with a GELU activation after each
layer.

\section{Signal Preprocessing and Backbone Details}
\label{app:preprocessing}

\textbf{CSI-ratio preprocessing.}
For each wireless link, we first form the CSI ratio between two receiving antennas, which suppresses phase distortions shared by the two RF chains and makes the resulting phase more suitable for sensing~\cite{zeng2019farsense}. 
We then extract the amplitude and phase of processed CSI and stack them over all $N_{\mathrm{sc}}$ valid subcarriers, producing the amplitude--phase input
\begin{equation}
    A_i \in \mathbb{R}^{2N_{\mathrm{sc}}\times 224}.
\end{equation}
Following the preprocessing pipeline of GesFi~\cite{zhang2026beyond}, we further process the CSI-ratio sequence to obtain the Doppler frequency spectrum (DFS). We retain 121 Doppler bins and temporally resample every gesture sample to 224 time steps. Thus, the two modalities contain $2N_{\mathrm{sc}}+121$ feature rows in total. For the 30-subcarrier CSI used in our evaluation, the AP and DFS inputs have sizes $60\times224$ and $121\times224$, respectively. 

\textbf{Signal encoders.}
The AP and DFS branches use separate ResNet-18 encoders~\cite{he2016deep}, both initialized with ImageNet-pretrained weights~\cite{deng2009imagenet}. Although the temporal dimension is standardized to 224, we do not resize either sensing representation to the conventional $224\times224$ image resolution. Instead, $A_i$ and $D_i$ are fed to their respective encoders at their native resolutions, preserving the original subcarrier and Doppler-bin organization. The resulting branch features are concatenated and projected into the common per-link representation $f_i$ as described in Eq.~\ref{eq:pair_feature}.

\begin{figure}[ht]
  \centering
  \vspace{-0.12in}
  \includegraphics[width=0.65\linewidth]{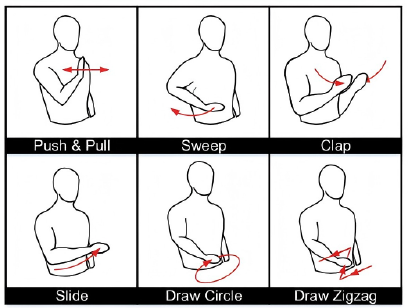}
   \vspace{-0.1in}
  \caption{Qualitative execution sequences of the six Widar3.0 gestures~\cite{zhang2021widar3}.}
  \label{fig:gesture_execution}
  \vspace{-0.15in}
\end{figure}

\section{Relation to Gesture Execution}
\label{app:gesture_execution}

Figure~\ref{fig:gesture_execution} summarizes the qualitative execution sequences of the six Widar3.0 gestures. These gestures are performed around the body and are not confined to the world-coordinate $xy$ plane used by the two-dimensional bistatic model. Moreover, the paths visualized in Fig.~\ref{fig:support_trajectories} are cumulative latent velocities expressed in a hypothesis-conditioned effective motion frame, rather than metrically recovered hand positions. Their absolute orientation, scale, and shape therefore need not coincide with the execution sketches.

The meaningful correspondence is instead temporal and morphological. Push--Pull contains successive motion phases with a reversal; Draw-O contains sustained turning; and Draw-Zigzag contains alternating sharp turns. Sweep and Slide produce relatively smooth excursions, while Clap contains a shorter and more direct convergence phase. These execution characteristics are consistent with the recurrent patterns in Fig.~\ref{fig:support_trajectories}. We emphasize that the two supports remain exchangeable effective motion components and should not be interpreted as particular hands or body parts. Accordingly, this comparison supports the presence of gesture-related latent motion organization, rather than trajectory reconstruction.

\end{document}